\def \VersionLong {} %
\def \VersionAuthor {}
\ifdefined\VersionAuthor
	\newcommand{\AuthorVersion}[1]{#1}
	\newcommand{\FinalVersion}[1]{}
\else
	\newcommand{\AuthorVersion}[1]{}
	\newcommand{\FinalVersion}[1]{#1}
\fi

\ifdefined \VersionLong
	\newcommand{\LongVersion}[1]{#1}
	\newcommand{\ShortVersion}[1]{}
\else
	\newcommand{\LongVersion}[1]{}
	\newcommand{\ShortVersion}[1]{#1}
\fi

\documentclass[conference]{IEEEtran}
\IEEEoverridecommandlockouts
\usepackage{amsthm}
\theoremstyle{plain}

\newtheorem{theorem}{Theorem}
\newtheorem{corollary}{Corollary}

\theoremstyle{definition}
\newtheorem{definition}{Definition}
\newtheorem{example}{Example}

\theoremstyle{remark}
\newtheorem{remark}{Remark}

\renewenvironment{proof}{\begin{IEEEproof}}{\end{IEEEproof}}
\usepackage[svgnames,table]{xcolor}
\usepackage{amsmath,amssymb}
\usepackage[utf8]{inputenc}
\usepackage[english]{babel}
\usepackage{csquotes}
\usepackage[ruled,vlined,linesnumbered]{algorithm2e}
	\SetKwInOut{Input}{input}
	\SetKwInOut{Output}{output}
\usepackage{subcaption}
\usepackage{paralist} %
\usepackage{adjustbox}
\usepackage{verbatim} %
\usepackage{multirow}

\newenvironment{ienumeration}
	{\ifdefined\VersionLong\begin{enumerate}\else\begin{inparaenum}[\itshape i\upshape)]\fi}
	{\ifdefined\VersionLong\end{enumerate}\else\end{inparaenum}\fi}

\newenvironment{oneenumeration}
	{\ifdefined\VersionLong\begin{enumerate}\else\begin{inparaenum}[1)]\fi}
	{\ifdefined\VersionLong\end{enumerate}\else\end{inparaenum}\fi}

\ifdefined\VersionWithComments
	\usepackage{draftwatermark}
	\SetWatermarkText{draft}
	\SetWatermarkScale{2}
	\SetWatermarkColor[gray]{0.9}
\fi

\definecolor{USPNcobalt}{HTML}{293358}
\definecolor{USPNocre}{HTML}{8b7d6d}
\definecolor{USPNblanc}{HTML}{ffffff}
\definecolor{USPNceruleen}{HTML}{354878}
\definecolor{USPNsable}{HTML}{ad947e}

\usepackage[
		pdfauthor={Étienne André, Engel Lefaucheux and Dylan Marinho},%
		pdftitle={Expiring opacity problems in parametric timed automata},
		breaklinks  = true,
		colorlinks  = true,
		citecolor   = USPNsable,
		linkcolor   = USPNocre,
		urlcolor    = USPNceruleen,
	]{hyperref}
\usepackage[capitalise,english,nameinlink]{cleveref} %
\crefname{line}{\text{line}}{\text{lines}} %

\newcommand{\defProblem}[3]
{%
	\noindent\fcolorbox{black}{USPNsable!15}{
	\begin{minipage}{.95\columnwidth}
		\textbf{The #1\ifdefined\VersionAuthor{} problem\fi:}\\
		\textsc{Input}: #2\\
		\textsc{Problem}: #3
	\end{minipage}
	}
	\smallskip
}

\usepackage{tikz}
\usetikzlibrary{arrows,automata,positioning,shapes} %

\tikzstyle{pta}=[auto, ->, >=stealth']
\tikzstyle{every node}=[initial text=]
\tikzstyle{location}=[rectangle, rounded corners, minimum size=12pt, draw=black, fill=blue!10, inner sep=2pt]
\tikzstyle{invariant}=[draw=black, dotted, inner sep=1pt, node distance=0] %
\tikzstyle{final}=[double, fill=blue!50]

\tikzstyle{urgent}=[fill=yellow, thick, dotted] %
\tikzstyle{private}=[fill=red,thick]

\definecolor{coloract}{rgb}{0.50, 0.70, 0.30}
\definecolor{colorclock}{rgb}{0.4, 0.4, 1}
\definecolor{colordisc}{rgb}{1, 0, 1}
\definecolor{colorloc}{rgb}{0.4, 0.4, 0.65}
\definecolor{colorparam}{rgb}{1, 0.6, 0.0}

\definecolor{loccolor1}{rgb}{1, 0.3, 0.3}
\definecolor{loccolor2}{rgb}{0.3, 1, 0.3}
\definecolor{loccolor3}{rgb}{0.3, 0.3, 1}
\definecolor{loccolor4}{rgb}{1, 0.3, 1}
\definecolor{loccolor5}{rgb}{1, 1, 0.3}
\definecolor{loccolor6}{rgb}{0.3, 1, 1}
\definecolor{loccolor7}{rgb}{0.9, 0.6, 0.2}
\definecolor{loccolor8}{rgb}{0.7, 0.4, 1}
\definecolor{loccolor9}{rgb}{0.5, 1, 0.75}
\definecolor{loccolor10}{rgb}{0.8, 0.7, 0.6}
\definecolor{loccolor11}{rgb}{0.6, 0.7, 0.8}
\definecolor{loccolor12}{rgb}{0.2, 0.5, 0.9}
\definecolor{loccolor13}{rgb}{0.5, 0.9, 0.2}
\definecolor{loccolor14}{rgb}{0.9, 0.2, 0.5}
\definecolor{loccolor15}{rgb}{0.7, 0.7, 0.7}
\definecolor{loccolor16}{rgb}{0.8, 0.8, 0.5}

\newcommand{\styleclock}[1]{\ensuremath{\textcolor{colorclock}{{#1}}}}

\newcommand{\styleparam}[1]{\ensuremath{\textcolor{colorparam}{{#1}}}}
\newcommand{\cellHeader}[0]{\cellcolor{blue!20}\bfseries}

\newcommand{\cellYes}{\cellcolor{green!20}\textbf{$\surd$}}
\newcommand{\cellNo}{\cellcolor{red!20}\textbf{$\times$}}
\newcommand{\cellOpen}{\cellcolor{yellow!20}\textbf{$?$}}
\newcommand{\cellCref}[1]{{\scriptsize (\cref{#1})}}

\newcommand{\A}{\ensuremath{\mathcal{P}}}
\newcommand{\RA}[1]{\ensuremath{\mathcal{RA}_{#1}}}
\newcommand{\assign}{\leftarrow}

\newcommand{\class}[1]{\ensuremath{\left[#1\right]}}
\newcommand{\compOp}{\bowtie}

\newcommand{\init}{_0}
\newcommand{\integralp}[1]{\ensuremath{\lfloor#1\rfloor}}
\newcommand{\equivalent}{\ensuremath{\approx}}
\newcommand{\fract}[1]{\ensuremath{\text{fract}(#1)}}
\newcommand{\longuefleche}[1]{\stackrel{#1}{\longrightarrow}}
\newcommand{\longueflecheRel}[1]{\stackrel{#1}{\mapsto}}

\newcommand{\flecheRel}{{\rightarrow}}

\newcommand{\set}[1]{\ensuremath{\left\{#1\right\}}}
\newcommand{\Set}[2]{\ensuremath{\left\{#1 \ \vert\ #2 \right\}}}
\newcommand{\TA}{\ensuremath{\mathcal{A}}}

\newcommand{\Actions}{\ensuremath{\Sigma}}
\newcommand{\action}{\ensuremath{a}}
\newcommand{\actiontick}{\ensuremath{\mathsf{tick}}}

\newcommand{\Clock}{\mathbb{X}} %
\newcommand{\ClockCard}{H} %
\newcommand{\clock}{x} %
\newcommand{\clockval}{\mu} %
\newcommand{\ClocksZero}{\vec{0}}
\newcommand{\sclockx}{\ensuremath{\styleclock{x}}}

\newcommand{\duration}{\ensuremath{\mathit{dur}}}

\newcommand{\edge}{e}
\newcommand{\Edges}{E}
\newcommand{\extraAction}{\ensuremath{\sharp}}

\newcommand{\guard}{g}
\newcommand{\invariant}{I}

\newcommand{\loc}{\ensuremath{\ell}} %
\newcommand{\locinit}{\loc\init}
\newcommand{\Loc}{L} %
\newcommand{\locfinal}{\ensuremath{\loc_f}}
\newcommand{\locpriv}{\ensuremath{\loc_{\mathit{priv}}}}

\newcommand{\TTS}{\ensuremath{\mathcal{T}}}

\newcommand{\varrun}{\rho} %

\newcommand{\PTA}{\A}
\newcommand{\PTAextend}{\ensuremath{(\Actions, \Loc, \locinit, \locpriv, \locfinal, \Clock, \Param, \invariant, \Edges)}}

\newcommand{\Param}{\mathbb{P}} %
\newcommand{\param}{p} %
\newcommand{\parami}[1]{p_{#1}} %
\newcommand{\ParamCard}{M} %
\newcommand{\pval}{v} %

\newcommand{\concstate}{\ensuremath{s}} %
\newcommand{\state}{\ensuremath{s}}
\newcommand{\s}{\ensuremath{s}}
\newcommand{\States}{\ensuremath{\mathbb{S}}} %

\newcommand{\region}{\ensuremath{r}}
\newcommand{\Regions}{\ensuremath{{\mathcal{R}}}}
\newcommand{\regionEdges}{\ensuremath{\mathcal{F}}}
\newcommand{\RegionGraph}{\ensuremath{\mathcal{RG}}}

\newcommand{\resets}{R}

\newcommand{\reset}[2]{\ensuremath{[#1]_{#2}}}
\newcommand{\valuate}[2]{\ensuremath{#2(#1)}}
\newcommand{\setN}{\ensuremath{\mathbb N}}
\newcommand{\setNinf}{\ensuremath{\mathbb{N}^{\infty}}}
\newcommand{\setQ}{\ensuremath{{\mathbb Q}}}
\newcommand{\setQplus}{\ensuremath{\setQ_{+}}} %
\newcommand{\setR}{\ensuremath{\mathbb R}}
\newcommand{\setRplusinf}{\ensuremath{\setR^{\infty}_+}}
\newcommand{\setRgeqzero}{\ensuremath{\setR_+}}
\newcommand{\setZ}{\ensuremath{\mathbb Z}}

\newcommand{\expiringBound}{\ensuremath{\Delta}}
\newcommand{\setBound}{\ensuremath{\mathcal{D}}} %
\newcommand{\boundReach}{\ensuremath{T}}

\newcommand{\runDuration}[1]{\ensuremath{\duration(#1)}}
\newcommand{\statesDuration}[3]{\ensuremath{\duration_{#1}(#2,#3)}}
\newcommand{\locationsDuration}[3]{\ensuremath{\duration_{#1}(#2,#3)}}

\newcommand{\PrivDurReach}[1]{\ensuremath{\mathit{DVisit}^\mathit{priv}(#1)}}
\newcommand{\PubDurReach}[1]{\ensuremath{D\overline{\mathit{Visit}}^{\mathit{priv}}(#1)}}

\newcommand{\PrivReach}[1]{\ensuremath{\mathit{Visit}^{\mathit{priv}}(#1)}}
\newcommand{\PubReach}[1]{\ensuremath{\overline{\mathit{Visit}}^{\mathit{priv}}(#1)}}

\newcommand{\TempSupPrivDurReach}[2]{\ensuremath{\mathit{DVisit}^\mathit{priv}_{>#2}(#1)}}
\newcommand{\TempInfPrivDurReach}[2]{\ensuremath{\mathit{DVisit}^\mathit{priv}_{\leq#2}(#1)}}
\newcommand{\TempSupPrivReach}[2]{\ensuremath{\mathit{Visit}^{\mathit{priv}}_{>#2}(#1)}}
\newcommand{\TempInfPrivReach}[2]{\ensuremath{\mathit{Visit}^{\mathit{priv}}_{\leq#2}(#1)}}

\newcommand{\ComplexityFont}[1]{{\sffamily\upshape #1}}

\newcommand{\TEXPTIME}{\ComplexityFont{3EXPTIME}\xspace}

\newcommand{\NP}{\ComplexityFont{NP}\xspace}

\newcommand{\NEXPTIME}{\ComplexityFont{NEXPTIME}\xspace}

\newcommand{\imitator}{\textsf{IMITATOR}}
\newcommand{\SpaceEx}{\textsc{SpaceEx}}

\newcommand{\tempOpaqueText}[1]{(\ensuremath{\leq}~\ensuremath{#1})-ET-opaque}
\newcommand{\tempOpacityText}[1]{(\ensuremath{\leq}~\ensuremath{#1})-ET-opacity}
\newcommand{\weakTempOpaqueText}[1]{weakly \tempOpaqueText{#1}}
\newcommand{\weakTempOpacityText}[1]{weak \tempOpacityText{#1}}
\newcommand{\fullTempOpaqueText}[1]{fully \tempOpaqueText{#1}}
\newcommand{\fullTempOpacityText}[1]{full \tempOpacityText{#1}}
\newcommand{\weakfullTempOpaqueText}[1]{fully (resp.\ weakly) \tempOpaqueText{#1}}
\newcommand{\weakfullTempOpacityText}[1]{full (resp.\ weak) \tempOpacityText{#1}}

\ifdefined \VersionWithComments
	\definecolor{colorok}{RGB}{80,80,150}
\else
	\definecolor{colorok}{RGB}{0,0,0}
\fi

\newcommand{\eg}{\textcolor{colorok}{e.g.,}\xspace} %

\newcommand{\ie}{\textcolor{colorok}{i.e.,}\xspace} %
\newcommand{\st}{\textcolor{colorok}{s.t.}\xspace}

\newcommand{\wrt}{{w.r.t.}\xspace} %

\AuthorVersion{%
	\usepackage[backend=biber,backref=true,style=alphabetic,url=false,doi=true,defernumbers=true,sorting=anyt,maxnames=99]{biblatex} %
	\addbibresource{PTA.bib}

	\renewbibmacro*{doi+eprint+url}{%
		\iftoggle{bbx:doi}
			{\color{black!40}\footnotesize\printfield{doi}}
			{}%
		\newunit\newblock
		\iftoggle{bbx:eprint}
			{\usebibmacro{eprint}}
			{}%
		\newunit\newblock
		\iftoggle{bbx:url}
			{\usebibmacro{url+urldate}}
			{}%
	}

}

\makeatletter
\def\orcidID#1{\smash{\href{https://orcid.org/#1}{\protect\raisebox{-1.25pt}{\protect\includegraphics{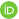}}}}}
\makeatother

\begin{document}

\pagestyle{plain}

\title{Expiring opacity problems in parametric timed automata\thanks{%
	\AuthorVersion{%
		This is the author (and slightly modified) version of the manuscript of the same name published in the proceedings of the 27th International Conference on Engineering of Complex Computer Systems (\href{https://www.irit.fr/iceccs2023/}{ICECCS 2023}).
		The published version is available at \href{https://www.doi.org/10.1109/ICECCS59891.2023.00020}{10.1109/ICECCS59891.2023.00020}.
		Modifications in this manuscript include a few minor modified notations, to be consistent with our invited paper at TiCSA~2023~\cite{ALLMS23} published after this paper.
	}%
	This work is partially supported by the ANR-NRF French-Singaporean research program ProMiS (ANR-19-CE25-0015 / 2019 ANR NRF 0092)
		and
		the ANR research program BisoUS (ANR-22-CE48-0012).
}
}
\author{
	\IEEEauthorblockN{\'Etienne Andr\'e\orcidID{0000-0001-8473-9555}}
	\IEEEauthorblockA{\textit{Université Sorbonne Paris Nord, LIPN, CNRS UMR 7030}\\
		\textit{F-93430 Villetaneuse, France}}
	\and
	\IEEEauthorblockN{Engel Lefaucheux\orcidID{0000-0003-0875-300X}}
	\IEEEauthorblockN{Dylan Marinho\orcidID{0000-0002-2548-6196}}
	\IEEEauthorblockA{\textit{Université de Lorraine, CNRS, Inria, LORIA}\\
		\textit{F-54000 Nancy, France}
	}
}

\maketitle

\begin{abstract}
	Information leakage can have dramatic consequences on the security of real-time systems.
	Timing leaks occur when an attacker is able to infer private behavior depending on timing information.
	In this work, we propose a definition of expiring timed opacity \wrt{} execution time, where a system is opaque whenever the attacker is unable to deduce the reachability of some private state solely based on the execution time; in addition, the secrecy is violated only when the private state was entered ``recently'', \ie{} within a given time bound (or expiration date) prior to system completion.
	This has an interesting parallel with concrete applications, notably cache deducibility: it may be useless for the attacker to know the cache content too late after its observance.
	We study here expiring timed opacity problems in timed automata.
	We consider the set of time bounds (or expiration dates) for which a system is opaque and show when they can be effectively computed for timed automata.
	We then study the decidability of several parameterized problems, when not only the bounds, but also some internal timing constants become timing parameters of unknown constant values.
\end{abstract}


\section{Introduction}\label{section:introduction}
Complex timed systems combine hard real-time constraints with concurrency.
Information leakage can have dramatic consequences on the security of such systems.
Among harmful information leaks, the \emph{timing information leakage} is the ability for an attacker to deduce internal information depending on timing information.
In this work, we focus on timing leakage through the total execution time, \ie{} when a system works as an almost black-box, and the ability of the attacker is limited to know the model and observe the total execution time.
We consider the setting of timed automata (TAs), which is a popular extension of finite-state automata with clocks~\cite{AD94}.

\paragraph{Context and related works}
Franck~Cassez proposed in~\cite{Cassez09} a first definition of timed opacity: the system is opaque if an attacker cannot deduce whether some set of actions was performed, by only observing a given set of observable actions together with their timestamp.
It is then proved in~\cite{Cassez09} that it is undecidable whether a TA is opaque, even for the restricted class of event-recording automata~\cite{AFH99} (a subclass of TAs).
This notably relates to the undecidability of timed language inclusion for TAs~\cite{AD94}.
The aforementioned negative result leaves hope only if the definition or the setting is changed, which was done in three main lines of works.
First, in~\cite{WZ18,WZA18}, the input model is simplified to \emph{real-time automata}, a severely restricted formalism compared to TAs.
Timed aspects are only considered by interval restrictions over the total elapsed time along transitions.
Real-time automata can be seen as a subclass of TAs with a single clock, reset at each transition.
In this setting, (initial-state) opacity becomes decidable~\cite{WZ18,WZA18}.

Second, in~\cite{ALMS22}, we consider a less powerful attacker, who has access only to the \emph{execution time}: this is \emph{execution-time opacity (ET-opacity)}.\footnote{%
	In~\cite{ALMS22}, this notion was only referred to as ``timed opacity''.
}
In the setting of TAs, the execution time denotes the time from the system start to the reachability of a given (final) location.
Therefore, given a secret location, a TA is ET-opaque for an execution time~$d$ if there exist at least two runs of duration~$d$ from the initial location to a final location: one visiting the secret location, and another one \emph{not} visiting the secret location.
Deciding whether at least one such~$d$ exists can be seen as an \emph{existential} version of ET-opacity---which we do not consider here. %
Then, the system is \emph{fully ET-opaque}
	if it is ET-opaque \emph{for all execution times}: that is, for each possible~$d$, either the final location is not reachable at all, or the final location is reachable for at least two runs, one visiting the secret location, and another one not visiting it.
These two definitions of ET-opacity are shown to be decidable for TAs~\cite{ALMS22}.
We also studied various parametric extensions, and notably showed that the parametric emptiness problem (the emptiness over the parameter valuations set for which the TA is existentially-ET-opaque) becomes decidable for a subclass of parametric timed automata (PTAs)~\cite{AHV93}, called L/U-PTAs~\cite{HRSV02}, where parameters are partitioned between lower-bound and upper-bound parameters.

Third, in~\cite{AEYM21}, the authors consider a time-bounded notion of the opacity of~\cite{Cassez09}, where the attacker has to disclose the secret before an upper bound, using a partial observability.
This can be seen as a secrecy with an \emph{expiration date}.
The rationale is that retrieving a secret ``too late'' is useless; this is understandable, \eg{} when the secret is the value in a cache; if the cache was overwritten since, then knowing the secret is probably useless in most situations.
In addition, the analysis is carried over a time-bounded horizon; this means there are two time bounds in~\cite{AEYM21}: one for the secret expiration date, and one for the bounded-time execution of the system.
(We consider only the former one in this work, and lift the assumption regarding the latter.)
The authors prove that deciding whether a system is time-bounded opaque under a bounded time horizon, with a notion close to our weakness definition (unidirectional language inclusion), is decidable for~TAs. %
A construction and an algorithm are also provided to solve it; a case study is verified using \SpaceEx{}~\cite{FLDCRLRGDM11}.

In an orthogonal line of works, \cite{BT03,AK20} consider \emph{non-interference} for (parametric) TAs, allowing to quantify the frequency of an attack; this can be seen as a measure of the strength of an attack, depending on the frequency of the admissible actions.
Also see \cite{AA23survey} for a survey on security problems in timed automata.

\paragraph{Contribution}
In this work, we consider an \emph{expiring version of ET-opacity}, where the secret is subject to an expiration date;
this can be seen as a combination of both concepts from~\cite{ALMS22,AEYM21}.
That is, we consider that an attack is successful only when the attacker can decide that the secret location was entered less than $\expiringBound$ time units before the system completion.
Conversely, if the attacker exhibits an execution time~$d$ for which it is certain that the secret location was visited, but this location was entered strictly more than~$\expiringBound$ time units prior to the system completion, then this attack is useless, and can be seen as a failed attack.
The system is therefore fully expiring ET-opaque if the set of execution times for which the private location was entered within $\expiringBound$ time units prior to system completion is exactly equal to the set of execution times for which the private location was either not visited or entered $> \expiringBound$ time units prior to system completion.

In addition, when the former (secret) set of times is \emph{included} into the latter (non-secret) set of times, we say that the system is \emph{weakly} expiring ET-opaque; this encodes situations when the attacker might be able to deduce that no secret location was visited, but is not able to confirm that the secret location \emph{was} indeed visited.

On the one hand, our attacker model is \emph{less powerful} than~\cite{AEYM21}, because our attacker has only access to the execution time (and to the input model); in that sense, our attacker capability is identical to~\cite{ALMS22}.
On the other hand, we lift the time-bounded horizon analysis from~\cite{AEYM21}, allowing to analyze systems without any assumption on their execution time; therefore, we only import from~\cite{AEYM21} the notion of \emph{expiring secret}.
Also note that our formalism is much more expressive (and therefore able to encode richer applications) than in~\cite{WZ18,WZA18} because we consider the full class of TAs instead of the restricted real-time automata.
We also consider parametric extensions, not discussed in~\cite{Cassez09,WZ18,WZA18,AEYM21}.

We first consider ET-opacity for TAs.
We show that it is possible to:
\begin{oneenumeration}%
	\item decide whether a TA is fully (resp.\ weakly) expiring ET-opaque for a given time bound~$\expiringBound$ (decision problem);
	\item decide whether a TA is fully (resp.\ weakly) expiring ET-opaque for at least one bound~$\expiringBound$ (emptiness problem);
	\item compute the set of time bounds (or expiration dates) for which a TA is weakly expiring ET-opaque (computation problem).
\end{oneenumeration}%

Second, we show that, in PTAs, the emptiness of the parameter valuation sets for which the system is fully (resp.\ weakly) expiring ET-opaque is undecidable, even for the L/U-PTA subclass of PTAs, so far known for its decidability results.

\paragraph{Outline}
We recall preliminaries in \cref{section:preliminaries}.
We define expiring opacity problems in \cref{section:problems}.
We address problems for TAs in \cref{section:TAs}, and parametric extensions in \cref{section:PTAs}.
We conclude in \cref{section:conclusion}.

\section{Preliminaries}\label{section:preliminaries}

Let $\setN$, $\setZ$, $\setQplus$, $\setRgeqzero$, $\setR$ denote the sets of non-negative integers, integers, non-negative rational numbers, non-negative real numbers, and real numbers, respectively.
Let $\setNinf = \setN \cup \{+\infty\}$ and
$\setRplusinf = \setRgeqzero \cup \{+\infty\}$.

\AuthorVersion{%
 \subsection{Clocks and guards}\label{ss:clocks}
}

We assume a set~$\Clock = \{ \clock_1, \dots, \clock_\ClockCard \} $ of \emph{clocks}, \ie{} real-valued variables that all evolve over time at the same rate.
A clock valuation is a function
$\clockval : \Clock \rightarrow \setRgeqzero$.
We write $\ClocksZero$ for the clock valuation assigning $0$ to all clocks.
Given $d \in \setRgeqzero$, $\clockval + d$ denotes the valuation \st{} $(\clockval + d)(\clock) = \clockval(\clock) + d$, for all $\clock \in \Clock$.
Given $\resets \subseteq \Clock$, we define the \emph{reset} of a valuation~$\clockval$, denoted by $\reset{\clockval}{\resets}$, as follows: $\reset{\clockval}{\resets}(\clock) = 0$ if $\clock \in \resets$, and $\reset{\clockval}{\resets}(\clock)=\clockval(\clock)$ otherwise.

We assume a set~$\Param = \{ \param_1, \dots, \param_\ParamCard \} $ of \emph{parameters}, \ie{} unknown constants.
A parameter {\em valuation} $\pval$ is a function
$\pval : \Param \rightarrow \setQplus$.

A \emph{clock guard}~$\guard$ is a constraint over $\Clock \cup \Param$ defined by a conjunction of inequalities of the form
$\clock \compOp \sum_{1 \leq i \leq \ParamCard} \alpha_i \param_i + d$
	with $\clock \in \Clock$, $\param_i \in \Param$, $\alpha_i , d \in \setZ$
	and ${\compOp} \in \{<, \leq, =, \geq, >\}$.
Given~$\guard$, we write~$\clockval \models \valuate{\guard}{\pval}$ if %
the expression obtained by replacing each~$\clock$ with~$\clockval(\clock)$ and each~$\param$ with~$\pval(\param)$ in~$\guard$ evaluates to true. %

\LongVersion{%
\subsection{Parametric timed automata}
}

Parametric timed automata (PTAs) extend timed automata with parameters within guards and invariants in place of integer constants~\cite{AHV93}.
We extend PTAs with a special location called ``private location''.

\begin{definition}[PTA]\label{def:uPTA}
	A PTA $\PTA$ is a tuple \mbox{$\A = \PTAextend$}, where:
	\begin{ienumeration}
		\item $\Actions$ is a finite set of actions,
		\item $\Loc$ is a finite set of locations,
		\item $\locinit \in \Loc$ is the initial location,
		\item $\locpriv \in \Loc$ is the private location,
		\item $\locfinal \in \Loc$ is the final location,
		\item $\Clock$ is a finite set of clocks,
		\item $\Param$ is a finite set of parameters,
		\item $\invariant$ is the invariant, assigning to every $\loc\in \Loc$ a clock guard $\invariant(\loc)$,
		\item $\Edges$ is a finite set of edges  $\edge = (\loc,\guard,\action,\resets,\loc')$
		where~$\loc,\loc'\in \Loc$ are the source and target locations, $\action \in \Actions$, $\resets\subseteq \Clock$ is a set of clocks to be reset, and $\guard$ is a clock guard.
	\end{ienumeration}
\end{definition}

\LongVersion{%

 \subsection{Timed automata}\label{sec:TA-definition}
}

Given a PTA~$\A$ and\LongVersion{ a parameter valuation}~$\pval$, we denote by $\valuate{\A}{\pval}$ the non-parametric structure where all occurrences of a parameter~$\param_i$ have been replaced by~$\pval(\param_i)$.
\LongVersion{%
	We denote as a \emph{timed automaton} any structure $\valuate{\A}{\pval}$, by assuming a rescaling of the constants: by multiplying all constants in $\valuate{\A}{\pval}$ by the least common multiple of their denominators, we obtain an equivalent (integer-valued) TA\LongVersion{, as defined in~\cite{AD94}}.
}\ShortVersion{$\valuate{\A}{\pval}$ is a \emph{timed automaton}~\cite{AD94}.}

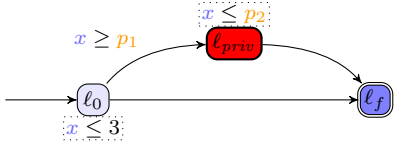
\begin{figure}[tb]

	\centering
	\footnotesize

	\begin{tikzpicture}[pta, scale=1, xscale=2.5, yscale=1.5]

		\node[location, initial] at (0, -.5) (s0) {$\loc_0$};

		\node[location, private] at (.75, 0) (s2) {$\locpriv$};

		\node[location, final] at (1.5, -.5) (s1) {$\locfinal$};

		\node[invariant, below=of s0] {$\sclockx \leq 3$};
		\node[invariant, above=of s2] {$\sclockx \leq \styleparam{\param_2}$};

		\path (s0) edge[bend left] node[align=center]{$\sclockx \geq \styleparam{\param_1}$} (s2);
		\path (s0) edge[] (s1);
		\path (s2) edge[bend left] node[align=center]{} (s1);

	\end{tikzpicture}
	\caption{A PTA example}
	\label{figure:example-PTA}

\end{figure}
\begin{example}\label{example:TA}
	Consider the PTA in \cref{figure:example-PTA} (inspired by \cite[Fig.~1b]{GMR07}), using one clock~$\clock$ and two parameters~$\param_1$ and~$\param_2$.
	$\loc_0$~is the initial location, while~$\locfinal$ is the final location.
\end{example}

\AuthorVersion{%
\subsubsection{Concrete semantics of TAs}
}
\begin{definition}[Semantics of a TA]
	Given a PTA $\PTA = \PTAextend$,
	and a parameter valuation~\(\pval\),
	the semantics $\TTS_{\valuate{\A}{\pval}}$ of $\valuate{\A}{\pval}$ is given by the timed transition system (TTS) $(\States, \s\init, \flecheRel)$, with
	\begin{itemize}
		\item $\States = \{ (\loc, \clockval) \in \Loc \times \setRgeqzero^\ClockCard \mid \clockval \models \valuate{\invariant(\loc)}{\pval} \}$, %
		\LongVersion{\item} $\s\init = (\locinit, \ClocksZero) $,
		\item  $\flecheRel$ consists of the discrete and (continuous) delay transition relations:
		\begin{ienumeration}
			\item discrete transitions: $(\loc,\clockval) \longueflecheRel{\edge} (\loc',\clockval')$, %
			if $(\loc, \clockval) , (\loc',\clockval') \in \States$, and there exists $\edge = (\loc,\guard,\action,\resets,\loc') \in \Edges$, such that $\clockval'= \reset{\clockval}{\resets}$, and $\clockval \models \valuate{\guard}{\pval}$.
			\item delay transitions: $(\loc,\clockval) \longueflecheRel{d} (\loc, \clockval+d)$, with $d \in \setRgeqzero$, if $\forall d' \in [0, d], (\loc, \clockval+d') \in \States$.
		\end{ienumeration}
	\end{itemize}
\end{definition}

Moreover we write $(\loc, \clockval)\longuefleche{(d, \edge)} (\loc',\clockval')$ for a combination of a delay and a discrete transition if
$\exists  \clockval'' :  (\loc,\clockval) \longueflecheRel{d} (\loc,\clockval'') \longueflecheRel{\edge} (\loc',\clockval')$.

Given a TA~$\valuate{\A}{\pval}$ with concrete semantics $(\States, \s\init, \flecheRel)$, we refer to the states of~$\States$ as the \emph{concrete states} of~$\valuate{\A}{\pval}$.
A \emph{run} of~$\valuate{\A}{\pval}$ is an alternating sequence of concrete states of $\valuate{\A}{\pval}$ and pairs of edges and delays starting from the initial state $\s\init$ of the form
$\state_0, (d_0, \edge_0), \state_1, \cdots$
with
$\dots$, $\edge_i \in \Edges$, $d_i \in \setRgeqzero$ and
$\state_i \longuefleche{(d_i, \edge_i)} \state_{i+1}$ for $i = 1, 2, \cdots$.
The \emph{duration} between two states of a finite run $\varrun : \state_0, (d_0, \edge_0), \state_1, \cdots, \state_k$ is $\statesDuration{\varrun}{\state_i}{\state_j} = \sum_{i \leq m \leq j-1} d_m$.
The \emph{duration} of a finite run $\varrun : \state_0, (d_0, \edge_0), \state_1, \cdots, \state_k$ is $\runDuration{\varrun} = \statesDuration{\varrun}{\state_0}{\state_k} = \sum_{0 \leq m \leq k-1} d_m$.
We also define the \emph{duration} between two locations $\loc_1$ and $\loc_2$ as the duration $\locationsDuration{\varrun}{\loc_1}{\loc_2} = \statesDuration{\varrun}{\state_i}{\state_j}$ with
$\varrun : \state_0, (d_0, \edge_0), \state_1, \cdots, \state_i, \cdots, \state_j, \cdots, \state_k$ where
$\state_j$ the first occurrence of a state with location $\loc_2$
and
$\state_i$ is the last state of~$\varrun$ with location $\loc_1$ before~$\state_j$.
We choose this definition to coincide with the definitions of opacity that we will define later (\cref{definition:temporary-timed-opacity}). Indeed, we want to make sure that revealing a secret ($\loc_1$ in this definition) is not a failure if it is done after a given time. Thus, as soon as the system reaches its final state ($\loc_2$), we will be interested in knowing how long the secret has been present, and thus the last time it was entered~($\state_i$).

\begin{example}\label{example:TA:runs}
	Consider again the PTA in \cref{figure:example-PTA}.
	Let $\pval$ be such that $\pval(\param_1) = 1$ and $\pval(\param_2) = 2$.
	Consider the following run~$\varrun$ of $\valuate{\PTA}{\pval}$:
	$(\loc_0, \clock = 0) , (1.4, \edge_2) , (\locpriv, \clock = 1.4) , (0.4, \edge_3) , (\locfinal, \clock = 1.8)$, where
	$\edge_2$ is the edge from $\loc_0$ to~$\locpriv$ in \cref{figure:example-PTA},
	and
	$\edge_3$ is the edge from $\locpriv$ to~$\locfinal$.
	We write ``$\clock = 1.4$'' instead of ``$\clockval$ such that $\clockval(\clock) = 1.4$''.
	We have $\duration(\varrun) = 1.4 + 0.4 = 1.8$ and $\locationsDuration{\varrun}{\locpriv}{\locfinal} = 0.4$.
\end{example}

\AuthorVersion{
\subsubsection{Timed automata regions}\label{sss:regions}
}

Let us now recall the concept of \AuthorVersion{regions and }the region graph~\cite{AD94}.

Given a TA~$\TA$, for a clock $\clock_i$, we denote by $c_i$ the largest constant to which $\clock_i$ is compared within the guards and invariants of~$\TA$ (that is, $c_i = \max_{i}( \{\ d_i \mid \clock \compOp d_i \text{~appears in a guard or invariant of~}\TA \}$).
Given $\alpha\in\setR$, let $\integralp{\alpha}$ and $\fract{\alpha}$ denote respectively the integral part and the fractional part of~$\alpha$.

\begin{example}\label{example:largestConstant}
	Consider again the PTA in \cref{figure:example-PTA}, and let $\pval$ be such that $\pval(\param_1) = 2$ and $\pval(\param_2) = 4$.
	In the TA $\valuate{\A}{\pval}$, the clock $\clock$ is compared to the constants in $\set{2, 3, 4}$.
	In that case, $c = 4$ is the largest constant to which the clock $\clock$ is compared.
\end{example}
\begin{definition}[Region equivalence]\label{definition:region-equivalence}
	We say that two clock valuations $\clockval$ and $\clockval'$ are equivalent,
	denoted $\clockval \equivalent \clockval'$,
	if the following three conditions hold for any pair of clocks $x_i,x_j$:
	\begin{enumerate}
		\item $\integralp{\clockval(\clock_i)} = \integralp{\clockval'(\clock_i)}$
			or $\big( \clockval(\clock_i)>c_i \text{ and } \clockval'(\clock_i)>c_i \big)$;
		\item $\fract{\clockval(\clock_i)}\leq\fract{\clockval(\clock_j)}$ iff $\fract{\clockval'(\clock_i)}\leq\fract{\clockval'(\clock_j)}$; and
		\item $\fract{\clockval(\clock_i)}=0$ iff $\fract{\clockval'(\clock_i)}=0$.
	\end{enumerate}
\end{definition}
The equivalence relation $\equivalent$ is extended to the states of $\TTS_{\TA}$: if ${\concstate = (\loc,\clockval)}, {\concstate' = (\loc',\clockval')}$ are two states of $\TTS_{\TA}$, we write ${\concstate\equivalent \concstate'}$ iff ${\loc=\loc'}$ and ${\clockval\equivalent \clockval'}$.

We denote by $\class{\concstate}$ the equivalence class of $\concstate$ for $\equivalent$. %
A \emph{region} is an equivalence class $\class{\concstate}$ of $\equivalent$.
The set of all regions is denoted by~$\Regions_{\TA}$.
Given a state $\concstate = (\loc, \clockval)$ and $d \geq 0$, we write $\concstate + d$ to denote $(\loc, \clockval + d)$.%

\begin{definition}[Region graph \cite{BDR08}]\label{definition:region-graph}
	The \emph{region graph} ${\RegionGraph_{\TA} = (\Regions_{\TA},\regionEdges_{\TA})}$ is a finite graph with:
	\begin{itemize}
		\item $\Regions_{\TA}$ as the set of vertices
		\item given two regions $\region=\class{\concstate},\region'=\class{\concstate'}\in\Regions_{\TA}$, we have $(\region,\region')\in\regionEdges_{\TA}$ if one of the following holds:
		\begin{itemize}
			\item $\concstate \longueflecheRel{\edge} \concstate' \in \TTS_{\TA}$ for some $\edge \in \Edges$ (\emph{discrete \AuthorVersion{instantaneous }transition});

			\item $\region'$ is a time successor of $\region$, \ie{} $r\neq r'$ and there exists $d$ such that $\concstate + d\in \region'$ and ${\forall d'<d, \concstate+d'\in \region\cup\region'}$  (\emph{delay transition});

			\item $\region=\region'$ is unbounded, \ie{} $\concstate = (\loc,\clockval)$ with $\clockval(\clock_i)>c_i$ for all~$x_i$  (\emph{equivalent unbounded regions}).
		\end{itemize}
	\end{itemize}
\end{definition}

We now define a version of the region automaton based on~\cite{BDR08} where the only letter that can be read (``$\actiontick$'') means that one time unit has passed.
Note that this automaton is not timed.
As such, it is as usual described by a tuple $(\Sigma, Q, q_0,F,T)$ where $\Sigma$ is the alphabet, $Q$ is the set of states, $q_0$ is the initial state, $F$ is the set of final states and $T\in (Q\times \Sigma \times Q)$ is the set of transitions.

We assume that the original TA~$\TA$ possesses a special clock~$\clock_\actiontick$ that is always reset every 1 time unit (through appropriate invariants and resets).
This clock does not affect the behavior of the TA, but every time it is reset, we know that one unit of time passed.
We also assume that the TA is deadlocked once~$\locfinal$ is reached (\ie{} no transition can be taken and no time can elapse).

\begin{definition}[Region automaton~\cite{BDR08}]\label{definition:region-autmaton}
	The \emph{region automaton} of a TA~$\TA$ is $\RA{\TA} = \set{\set{\actiontick}, \Regions_{\TA}, \class{\state\init}, F, T}$
	where
	\begin{ienumeration}
		\item $\actiontick$ is the only action;
		\item $\Regions_{\TA}$ is the set of states (a state of $\RA{\TA}$ is a region of~$\TA$);
		\item $\class{\state\init}$ is the initial state (the region associated to the initial location of~$\TA$);
		\item the set of final states $F$ is the set of regions associated to the location $\locpriv$ where $\clock_\actiontick$ is not equal to~$1$ (\ie{} the set of regions $\region=\class{(\locpriv,\clockval)}$ where $\valuate{\clock_\actiontick}{\clockval}< 1$)
		\item $(\region, \action ,\region')\in T$ iff $(\region,\region')\in \regionEdges_{\TA}$ and $\action = \actiontick$ if $\clock_\actiontick$ was reset in the
		discrete instantaneous transition corresponding to $(\region,\region')$, and $\action = \varepsilon$ otherwise.
	\end{ienumeration}
\end{definition}

An important property of this automaton is that the word $\actiontick^k$ with $k\in \setN$ is accepted by $\RA{\TA}$ iff there exists a run reaching the final location of~$\TA$
 within $[k,k+1)$ time units.

\section{Expiring execution-time opacity problems}\label{section:problems}

\LongVersion{%
In this section, we formally introduce the problems we address in this paper.
In the following, let $\TA$ be a TA.

\subsection{Expiring execution-time opacity}
}

Given a TA~$\TA$ and a run~$\varrun$, we say that $\locpriv$ is \emph{reached on the way to~$\locfinal$ in~$\varrun$} if $\varrun$ is of the form $(\loc_0, \clockval_0), (d_0, \edge_0), (\loc_1, \clockval_1), \cdots, (\loc_m, \clockval_m), (d_m, \edge_m), \cdots (\loc_n, \clockval_n)$ %
for some~$m,n \in \setN$ such that $\loc_m = \locpriv$, $\loc_n = \locfinal$ and $\forall 0 \leq i \leq n-1, \loc_i \neq \locfinal$.
We denote by $\PrivReach{\TA}$ the set of those runs, and refer to them as \emph{private} runs. We denote by $\PrivDurReach{\TA}$ the set of all the durations of these runs.
Conversely, we say that
$\locpriv$ is avoided on the way to~$\locfinal$ in~$\varrun$
if $\varrun$ is of the form $(\loc_0, \clockval_0), (d_0, \edge_0), (\loc_1, \clockval_1), \cdots, (\loc_n, \clockval_n )$ %
with $\loc_n = \locfinal$ and $\forall 0 \leq i < n, \loc_i \notin \{\locpriv,\locfinal\}$.
We denote the set of those runs by~$\PubReach{\TA}$, referring to them as \emph{public} runs,
and by $\PubDurReach{\TA}$ the set of all the durations of these public runs. %

Given $\expiringBound \in \setRplusinf$, we define $\TempSupPrivReach{\TA}{\expiringBound}$ (resp.\ $\TempInfPrivReach{\TA}{\expiringBound}$)
 as the set of runs $\varrun\in\PrivReach{\TA}$
 \st{} $\locationsDuration{\varrun}{\locpriv}{\locfinal}>\expiringBound$ (resp.\ $\locationsDuration{\varrun}{\locpriv}{\locfinal}\leq\expiringBound$).
 We refer to the runs of $\TempInfPrivReach{\TA}{\expiringBound}$ as \emph{secret} runs.
$\TempSupPrivDurReach{\TA}{\expiringBound}$ (resp.\ $\TempInfPrivDurReach{\TA}{\expiringBound}$) is the set of all the durations of the runs in $\TempSupPrivReach{\TA}{\expiringBound}$ (resp.\ $\TempInfPrivReach{\TA}{\expiringBound}$).

We define below two notions of execution-time opacity \wrt{} a time bound~$\expiringBound$.
We will compare two sets:
\begin{ienumeration}%
	\item the set of execution times for which the private location was entered at most~%
		$\expiringBound$ time units prior to system completion; and
	\item the set of execution times for which either the private location was not visited at all, or it was last entered more than~%
$\expiringBound$ time units prior to system completion (which, in our setting, is somehow similar to \emph{not} visiting the private location, in the sense that entering it ``too early'' is considered of little interest).
\end{ienumeration}%
If both sets match, the system is \fullTempOpaqueText{\expiringBound}.
If the former is included into the latter, then the system is \weakTempOpaqueText{\expiringBound}.

\begin{definition}[\tempOpacityText{\expiringBound}]\label{definition:temporary-timed-opacity}
	Given a TA~$\TA$ and
	a bound (\ie{} an expiration date for the secret)~$\expiringBound\in\setRplusinf$
	we say that  $\TA$ is \emph{\fullTempOpaqueText{\expiringBound}} if
	$\TempInfPrivDurReach{\TA}{\expiringBound} = \TempSupPrivDurReach{\TA}{\expiringBound} \cup \PubDurReach{\TA}$.
	Moreover,
	we say that  $\TA$ is \emph{\weakTempOpaqueText{\expiringBound}} if
	$\TempInfPrivDurReach{\TA}{\expiringBound} \subseteq \TempSupPrivDurReach{\TA}{\expiringBound} \cup \PubDurReach{\TA}$.

\end{definition}

\begin{remark}\label{remark:weak}
Our notion of \emph{weak} opacity may still leak some information:
on the one hand, if a run indeed enters the private location $\leq \expiringBound$ time units before system completion, there exists an equivalent run not visiting it (or entering it earlier), and therefore the system is opaque;
\emph{but} on the other hand, there may exist execution times for which the attacker can deduce that the private location was \emph{not} entered $\leq \expiringBound$ before system completion.
This remains acceptable in some cases, and this motivates us to define a weak version of \tempOpacityText{\expiringBound}.
Also note that the ``initial-state opacity'' for real-time automata considered in~\cite{WZ18} can also be seen as \emph{weak} in the sense that their language inclusion is also unidirectional.
\end{remark}

\begin{example}\label{ex:defOpacity}
	Consider again the PTA in~\cref{figure:example-PTA}; let $\pval$ be such that $\pval(\param_1) = 1$ and $\pval(\param_2) = 2.5$.
	Fix $\expiringBound = 1$.

	We have:
	\begin{itemize}
		\item $\PubDurReach{\valuate{\PTA}{\pval}} = [0, 3] $
		\item $\TempSupPrivDurReach{\valuate{\PTA}{\pval}}{\expiringBound} = (2,2.5] $
		\item $\TempInfPrivDurReach{\valuate{\PTA}{\pval}}{\expiringBound} = [1,2.5] $
	\end{itemize}
	Therefore, we say that $\valuate{\PTA}{\pval}$ is: %
	\begin{itemize}
		\item \weakTempOpaqueText{1}, as $[1,2.5] \subseteq \big( (2,2.5] \cup [0,3] \big)$
		\item not \fullTempOpaqueText{1}, as $[1,2.5] \neq \big( (2,2.5] \cup [0,3] \big)$
	\end{itemize}

	As introduced in~\cref{remark:weak}, despite the \weakTempOpacityText{1} of~$\TA$, the attacker can deduce some information about the visit of the private location for some execution times.
	For example, if a run has a duration of 3 time units, it cannot be a private run, and therefore the attacker can deduce that the private location was not visited at all.
\end{example}

We define three different problems:

\defProblem
{\weakfullTempOpacityText{\expiringBound} decision}
{A TA~$\TA$ and a bound $\expiringBound\in\setRplusinf$}
{Decide whether $\TA$ is \weakfullTempOpaqueText{\expiringBound}}%

\defProblem
{\weakfullTempOpacityText{\expiringBound} emptiness}
{A TA~$\TA$}
{Decide the emptiness of the set of bounds $\expiringBound$ such that $\TA$ is \weakfullTempOpaqueText{\expiringBound}}

\defProblem
{\weakfullTempOpacityText{\expiringBound} computation}
{A TA~$\TA$}
{Compute the maximal set $\setBound$ of bounds such that $\TA$ is \weakfullTempOpaqueText{\expiringBound} for all $\expiringBound\in\setBound$}

\begin{example}\label{ex:opacityProblemsTA}
	Consider again the PTA in~\cref{figure:example-PTA}; let $\pval$ be such that $\pval(\param_1) = 1$ and $\pval(\param_2) = 2.5$ (as in \cref{ex:defOpacity}).
	Given $\expiringBound = 1$, the \weakTempOpacityText{\expiringBound} decision problem asks whether $\valuate{\PTA}{\pval}$ is \weakTempOpaqueText{\expiringBound}---the answer is ``yes'' from \cref{ex:defOpacity}.
	The \weakTempOpacityText{\expiringBound} emptiness problem is therefore ``no'' because the set of bounds $\expiringBound$ such that $\valuate{\PTA}{\pval}$ is \weakTempOpaqueText{\expiringBound} is not empty.
	Finally, the \weakTempOpacityText{\expiringBound} computation problem asks to compute all the corresponding bounds: in this example, the solution is ${\expiringBound\in\setRgeqzero}$.
\end{example}

Note that, when considering $\expiringBound = +\infty$,
$\TempSupPrivDurReach{\TA}{\expiringBound}=\emptyset$ and all the execution times of runs passing by $\locpriv$ are in $\TempInfPrivDurReach{\TA}{\expiringBound}$.
Therefore, \fullTempOpacityText{+\infty} matches the full ET-opacity\footnote{Named ``full timed opacity'' in~\cite{ALMS22}.} defined in~\cite{ALMS22}.
We can therefore notice that answering the \fullTempOpacityText{+\infty} decision problem is decidable (\cite[Proposition 5.3]{ALMS22}). %
However, the emptiness and computation problems cannot be reduced to full ET-opacity problems from~\cite{ALMS22}.
Conversely, it is possible to answer the full ET-opacity decision problem by checking the \fullTempOpacityText{+\infty} decision problem.
Moreover, full ET-opacity computation problem
reduces to \fullTempOpacityText{\expiringBound} computation: if $+\infty \in \setBound$,
we get the answer.

Note that our problems are incomparable to the ones addressed in~\cite{AEYM21} as the models used in their paper have a bounded execution time $<+\infty$, in addition to the bounded opacity~$\expiringBound$.

\section{Expiring execution-time opacity in TAs}\label{section:TAs}

\AuthorVersion{%
In this section, we consider the three problems defined previously on TAs.
}%
In general, the link between the full and weak notions of the three aforementioned problems is not obvious.
However, for a fixed value of $\expiringBound$, we establish the following theorem.

\begin{theorem}
\label{th:weaktonorm}
The \fullTempOpacityText{\expiringBound} decision problem reduces to the \weakTempOpacityText{\expiringBound} decision problem.
\end{theorem}
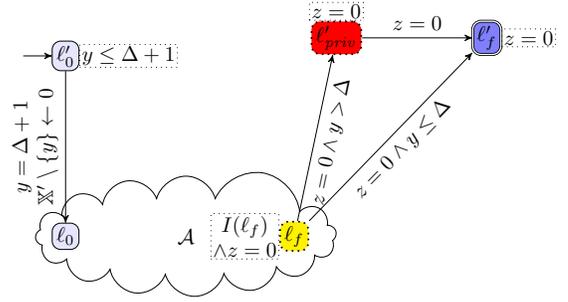
\begin{figure}[tb]
	\centering

	\scalebox{.8}
	{
		\begin{tikzpicture}[pta, node distance=2cm, thin]

			\node[location] (l0) at (-2, 0) {$\locinit$};
			\node[location, urgent] (lf) at (+1.8, 0) {$\locfinal$};
			\node[cloud, cloud puffs=15.7, cloud ignores aspect, minimum width=5cm, minimum height=2cm, align=center, draw] (cloud) at (0cm, 0cm) {$\TA$};

			\node[location, initial] (l0') at (-2, 3) {$\locinit'$};
			\node[location, urgent, private] (lpriv') at (2.5, 3.3) {$\locpriv'$};
			\node[location, final] (lf') at (+5, 3.3) {$\locfinal'$};

			\node[invariant, right=of l0'] {$y\leq \expiringBound + 1$};
			\node[invariant, above=of lpriv'] {$z=0$};
			\node[invariant, left=of lf, align=center] {$\invariant(\locfinal)$\\$\wedge z = 0$};
			\node[invariant, right=of lf'] {$z = 0$};

			\path (l0') edge node[above,align=center, rotate=90] {$y=\expiringBound+1$\\ $\Clock'\setminus\{y\}\assign 0$} (l0);
			\path (lf) edge node[below,align=center, sloped] {$z=0\wedge y>\expiringBound$} (lpriv');
			\path (lf) edge node[below,align=center, sloped] {$z=0\wedge y\leq\expiringBound$} (lf');
			\path (lpriv') edge node[above, align=center] {$z=0$} (lf');
		\end{tikzpicture}
	}

	\caption{Construction used in \cref{th:weaktonorm}}
	\label{figure:proof-weaktonorm}
\end{figure}
\begin{proof}
Fix a TA~$\TA$ and a time bound~$\expiringBound \in \setRplusinf$.
In this reduction, we build a new TA~$\TA'$ where secret and non-secret runs are swapped.
More precisely, we add a new clock~$y$ that measures how much time has elapsed since the latest entrance of the private location. It is thus reset whenever we enter the private location $\locpriv$.
This clock is initialized to value $\expiringBound +1$ (which can be ensured by waiting in a new initial location $\locinit'$ for $\expiringBound+1$ time units before going to the original initial location $\locinit$ and resetting every clock but~$y$).
When reaching the final location~$\locfinal$, one can urgently (a new clock~$z$ can be used to force the system to move immediately) move to a new secret location~$\locpriv'$
	if $y > \expiringBound$ and then to the new final location~$\locfinal'$; otherwise (if $y \leq \expiringBound$), the TA can go directly to the new final location~$\locfinal'$.
Therefore, a run that would not be secret (as $y > \expiringBound$) is now secret and reciprocally.
Then, by testing \weakTempOpacityText{\expiringBound} of both~$\TA$ and~$\TA'$, one can check \fullTempOpacityText{\expiringBound} of~$\TA$.

Formally, given a TA~$\TA = \ensuremath{(\Actions, \Loc, \locinit, \locpriv, \locfinal, \Clock, \invariant, \Edges)}$ and $\expiringBound \in \setRplusinf$,
we build a second TA $\TA' = \ensuremath{(\Actions\cup\{\extraAction\}, \Loc', \locinit', \locpriv', \locfinal', \Clock\cup\{y,z\}, \invariant', \Edges')}$ where
$\extraAction$ denotes a special action absent from~$\Actions$ and where:
\begin{itemize}
\item $\Loc'= \Loc \cup \{\locinit', \locpriv', \locfinal'\}$;
\item $\forall \loc \in \Loc \setminus \{\locfinal\} : \invariant'(\loc)=\invariant(\loc)$;
	$\invariant'(\locfinal) = (\invariant(\locfinal)\wedge z = 0$);
	$\invariant'(\locinit') = (y\leq \expiringBound + 1)$;
	$\invariant'(\locpriv') = (z=0)$;
	$\invariant'(\locfinal') = (z=0)$.
\item for each $(\loc,g,a,R,\loc')\in \Edges$, we add
$\big(\loc,g,a,R',\loc'\big)$ to $\Edges'$ where
	$R' = R\cup \{y,z\}$ if $\loc'=\locpriv$ and
	$R' = R\cup\{z\}$ otherwise. %
We also add the following edges to~$\Edges'$:
$\big\{ \big(\locinit', (y=\expiringBound+1), \extraAction, \Clock, \locinit\big) , \big(\locfinal, (z=0\wedge y> \expiringBound), \extraAction, \emptyset, \locpriv'\big) , \big(\locfinal, (z=0\wedge y\leq \expiringBound), \extraAction, \emptyset, \locfinal'\big) , \big(\locpriv', (z=0), \extraAction, \emptyset, \locfinal'\big) \big\}$.
\end{itemize}
We give a graphical representation of our construction in \cref{figure:proof-weaktonorm}.
There is a one-to-one correspondence between the secret (resp.\ non-secret) runs ending in $\locpriv$ in $\TA$
and the non-secret (resp.\ secret) runs ending in $\locpriv'$ in $\TA'$. Given $\rho$ a run
in $\TA$ and $\rho'$ the corresponding run in $\TA'$, then the duration of $\rho'$ is equal to the duration of $\rho $ plus $\expiringBound +1$ (the time waited in $\locinit'$).

Recall from \cref{definition:temporary-timed-opacity} the definition of \weakTempOpacityText{\expiringBound} for~$\A'$:
$\TempInfPrivDurReach{\TA'}{\expiringBound} \subseteq \TempSupPrivDurReach{\TA'}{\expiringBound} \cup \PubDurReach{\TA'}$.
\begin{enumerate}
	\item First consider the left-hand part ``$\TempInfPrivDurReach{\TA'}{\expiringBound}$'':
	these execution times correspond to runs of~$\TA'$ for which $\locpriv'$ was visited less than~$\expiringBound$ (and actually~0) time units prior to reaching~$\locfinal'$.
	These runs passed the $y > \expiringBound$ guard between~$\locfinal$ and~$\locpriv'$.
	From our construction, these runs correspond to runs of the original~$\TA$ either not passing at all by~$\locpriv$ (since $y$ was never reset since its initialization to~$\expiringBound + 1$, and therefore $y \geq \expiringBound + 1 > \expiringBound$), or to runs which visited~$\locpriv$ more than~$\expiringBound$ time units before reaching~$\locfinal$.
	Therefore,
	$\TempInfPrivDurReach{\TA'}{\expiringBound} = \Set{d + 1+\expiringBound}{d\in \TempSupPrivDurReach{\TA}{\expiringBound} \cup \PubDurReach{\TA}}$

	\item Second, consider the right-hand part ``$\TempSupPrivDurReach{\TA'}{\expiringBound} \cup \PubDurReach{\TA'}$'':
	the set $\TempSupPrivDurReach{\TA'}{\expiringBound}$ is necessarily empty, as any run of~$\TA'$ passing through~$\locpriv'$ reaches~$\locfinal'$ immediately in 0-time.
	The execution times from~$\PubDurReach{\TA'}$ correspond to runs of~$\A'$ \emph{not} visiting~$\locpriv'$, therefore for which only the guard $y \leq \expiringBound$ holds.
	Hence, they correspond to runs of~$\TA$ which visited $\locpriv$ less than~$\expiringBound$ time units prior to reaching~$\locfinal$.
	Therefore,
	$\TempSupPrivDurReach{\TA'}{\expiringBound} \cup \PubDurReach{\TA'} = \Set{d + 1+\expiringBound}{d\in \TempInfPrivDurReach{\TA}{\expiringBound}}$
\end{enumerate}
To conclude, checking that $\TA'$ is \weakTempOpaqueText{\expiringBound} (\ie{} $\TempInfPrivDurReach{\TA'}{\expiringBound} \subseteq \TempSupPrivDurReach{\TA'}{\expiringBound} \cup \PubDurReach{\TA'}$) is equivalent to $\TempSupPrivDurReach{\TA}{\expiringBound} \cup \PubDurReach{\TA} \subseteq \TempInfPrivDurReach{\TA}{\expiringBound}$.
Moreover, from \cref{definition:temporary-timed-opacity}, checking that $\TA$ is \weakTempOpaqueText{\expiringBound} denotes checking
$\TempInfPrivDurReach{\TA}{\expiringBound} \subseteq \TempSupPrivDurReach{\TA}{\expiringBound} \cup \PubDurReach{\TA}$.
Therefore, checking that both~$\TA'$ and~$\TA$ are \weakTempOpaqueText{\expiringBound} denotes $\TempInfPrivDurReach{\TA}{\expiringBound} = \TempSupPrivDurReach{\TA}{\expiringBound} \cup \PubDurReach{\TA}$, which is the definition of \fullTempOpacityText{\expiringBound} for~$\TA$.

To conclude, $\TA$ is \fullTempOpaqueText{\expiringBound} iff
$\TA$ and $\TA'$ are \weakTempOpaqueText{\expiringBound}.
\end{proof}

\begin{remark}
One can similarly establish the opposite reduction.
Given a TA~$\TA$, we build a second TA~$\TA'$ differing from the automaton
created in the proof of \cref{th:weaktonorm} only in the transitions exiting $\locfinal$:
the transitions exiting $\locfinal$ now are
$\{\big(\locfinal, (z=0\wedge y\leq \expiringBound), \extraAction, \emptyset, \locpriv'\big),
\big(\locfinal, (z=0\wedge y> \expiringBound), \extraAction, \emptyset, \locfinal'\big),
\big(\locfinal, (z=0\wedge y> \expiringBound), \extraAction, \emptyset, \locpriv'\big)\}$.
This construction ensures that the runs which were secret in~$\TA$ correspond to
secret runs of~$\TA'$, while the runs that were non-secret in~$\TA$ correspond to
a secret and a non-secret run of~$\TA'$.
Thus $\TempInfPrivDurReach{\TA'}{\expiringBound} \supseteq \TempSupPrivDurReach{\TA'}{\expiringBound} \cup \PubDurReach{\TA'}$ with equality iff
$\TempInfPrivDurReach{\TA}{\expiringBound} \subseteq \TempSupPrivDurReach{\TA}{\expiringBound} \cup \PubDurReach{\TA}$.
Therefore the \weakTempOpacityText{\expiringBound}
of~$\TA$ can be deduced from the \fullTempOpacityText{\expiringBound} of~$\TA'$.
\end{remark}

We now temporarily restrict $\expiringBound$ to the integer set~$\setNinf$.
(\cref{theorem-delta-setR} will lift the coming results to~$\setRplusinf$.)

\begin{theorem}
\label{th:TAtempop}
The \weakfullTempOpacityText{\expiringBound} decision problem is decidable in \NEXPTIME.
\end{theorem}
\begin{proof}
Given a TA~$\TA$, we first build two TAs from~$\TA$, named $\TA_p$ and $\TA_s$ and representing respectively
the public and secret behavior of the original TA, while each constant is multiplied by~2. %
The consequence of this multiplication is that the final location can be reached in time strictly
between $t$ and $t+1$ (with $t\in \setN$) by a public (resp.\ secret) run in $\TA$ iff the target can be reached in time $2t+1$ in the TA $\TA_p$ (resp.~$\TA_s$).
Note that the correctness of this statement is a direct consequence of~\cite[Lemma 5.5]{BDR08}. %

We then build the region automata
$\RA{p}$ and $\RA{s}$ (of $\TA_p$ and~$\TA_s$ respectively).%

$\RA{p}$ is a non-deterministic unary (the alphabet is restricted to a single letter) automaton with $\varepsilon$ transitions
the language of which is $\{\actiontick^k \mid \text{there is a run of duration } k \text{ in } \TA_p  \}$, and similarly for~$\RA{s}$.

We are interested in testing equality (resp.\ inclusion) of those languages
for deciding the \weakfullTempOpacityText{\expiringBound} decision problem.

\cite[Theorem~6.1]{SM73}
establishes that language equality of unary automata is \NP-complete and the same proof implies that inclusion is in \NP. As the region automata are exponential, we get the result.
\end{proof}

\begin{remark}
In~\cite{ALMS22}, we established that the \fullTempOpacityText{+\infty} decision problem is in \TEXPTIME.
\cref{th:TAtempop} thus extends our former results in three ways: by including the
parameter $\expiringBound$, by reducing the complexity and by considering as well the \emph{weak} notion of ET-opacity.
\end{remark}

\begin{theorem}
\label{th:weakcomp}
The \weakTempOpacityText{\expiringBound} computation problem is solvable.
\end{theorem}
\begin{proof}
First, we test whether $\TA$ is \weakTempOpaqueText{+\infty} thanks to
\cref{th:TAtempop}.
\begin{itemize}
	\item If $\TA$ is \weakTempOpaqueText{+\infty} then by definition (and monotonicity) of \weakTempOpacityText{\expiringBound}, $\TA$ is \weakTempOpaqueText{\expiringBound} for all $\expiringBound \in \setNinf$.

	\item Otherwise, there exists a duration $t\in \setRgeqzero$ such that $t\in \TempInfPrivDurReach{\TA}{+\infty} = \PrivDurReach{\TA}$ and $t\not \in \TempSupPrivDurReach{\TA}{+\infty} \cup \PubDurReach{\TA}=\PubDurReach{\TA}$. $t$ can be computed as a smallest word contradicting the inclusion of the language of the two exponential automata described in \cref{th:TAtempop}. Hence, $t$ is at most doubly exponential.
	For all $\expiringBound >t$, we thus have that
	$\TempInfPrivDurReach{\TA}{\expiringBound}\not\subseteq
	\TempSupPrivDurReach{\TA}{\expiringBound} \cup \PubDurReach{\TA}$ and thus that
	$\TA$ is not \weakTempOpaqueText{\expiringBound}.
	In order to synthesize the bounds $\expiringBound\in \setN$ such that
	$\TA$ is \weakTempOpaqueText{\expiringBound}, we therefore only have to test the finitely
	many integers below $t$ using \cref{th:TAtempop}.%
\end{itemize}
\end{proof}

\begin{corollary}\label{th:weakEmptiness}
The \weakTempOpacityText{\expiringBound} emptiness problem is decidable.
\end{corollary}
\begin{proof}
	According to~\cref{th:weakcomp}, \weakTempOpacityText{\expiringBound} computation is solvable.
	Therefore, to ask for emptiness, one can compute the set of bounds ensuring the \weakTempOpacityText{\expiringBound} and check its emptiness.
\end{proof}

In contrast to \weakTempOpacityText{\expiringBound} computation, we only show below that \fullTempOpacityText{\expiringBound} emptiness is decidable; the computation problem remains open.

\begin{theorem}\label{thm:opacityEmptiness}
The \fullTempOpacityText{\expiringBound} emptiness problem is decidable.
\end{theorem}
\begin{proof}
Given a TA $\TA$, using \cref{th:weakcomp}, we first compute the set of bounds
$\expiringBound$ such that $\TA$ is \weakTempOpaqueText{\expiringBound}.
As \fullTempOpacityText{\expiringBound} requires \weakTempOpacityText{\expiringBound}, if the
computed set is finite, then we only need to check the bounds of this set for
\fullTempOpacityText{\expiringBound} and thus synthesize all the bounds achieving
\fullTempOpacityText{\expiringBound}---which immediately allows us to decide emptiness.

If this set is infinite however, by the proof of \cref{th:weakcomp}, $\TA$ is \weakTempOpaqueText{\expiringBound} for any bound $\expiringBound\in\setNinf$.
To achieve \fullTempOpacityText{\expiringBound}, we only need to detect when the non-secret durations are included in the secret ones.
As the set of secret (resp.\ non-secret) durations increases (resp.\ decreases) when  $\expiringBound$ increases, there is a valuation of $\expiringBound$ achieving \fullTempOpacityText{\expiringBound} of $\TA$ iff $\TA$ is \fullTempOpaqueText{+\infty}.
The latter can be decided with \cref{th:TAtempop}.
\end{proof}

\begin{theorem}\label{theorem-delta-setR}
All aforementioned results with $\expiringBound \in \setNinf$ also hold for $\expiringBound \in \setRplusinf$.%
\end{theorem}
\begin{proof}%
\newcommand{\runprive}{\ensuremath{\varrun_\mathit{priv}}}%
\newcommand{\runpublic}{\ensuremath{\varrun_\mathit{pub}}}%
\newcommand{\timePriv}[1]{\ensuremath{\mathit{lt}_\mathit{priv}(#1)}}%
\newcommand{\timeFinal}[1]{\ensuremath{\mathit{lt}_f(#1)}}%
\newcommand{\VPriv}[1]{\ensuremath{V_\mathit{priv}(#1)}}%
\newcommand{\VFinal}[1]{\ensuremath{V_\mathit{f}(#1)}}%
\newcommand{\DPriv}[2]{\ensuremath{\mathit{DPriv}_{#1}(#2)}}%
\newcommand{\DPub}[2]{\ensuremath{\mathit{DPub}_{#1}(#2)}}%
\newcommand{\RRun}[1]{\ensuremath{\mathit{RRun}_{#1}}}%
Given a TA $\TA$ and $\expiringBound\in \setRplusinf\setminus\setNinf$, we will show that
$\TA$ is \weakfullTempOpaqueText{\expiringBound} iff it is \weakfullTempOpaqueText{
\integralp{\expiringBound} + \frac{1}{2}}.
Constructing the TA~$\TA'$ where every constant is multiplied by~2, we will thus have that
$\TA$ is \weakfullTempOpaqueText{\expiringBound} iff $\TA'$ is \weakfullTempOpaqueText{\expiringBound'} where
$\expiringBound' = 2\expiringBound$ if $\expiringBound \in \setN$
	and
$\expiringBound' = 2\integralp{\expiringBound}+1$ otherwise. The previous results of this section applying on $\TA'$, they can be transposed to $\TA$.

We now move to the proof that $\TA$ is \weakfullTempOpaqueText{\expiringBound} iff it is \weakfullTempOpaqueText{\integralp{\expiringBound} + \frac{1}{2}}.
Let $\expiringBound\in \setRgeqzero\setminus\setN$ such that $\TA$ is
\weakfullTempOpaqueText{\expiringBound} and let $\expiringBound' = \integralp{\expiringBound} + \frac{1}{2}$.

Given a run $\varrun\in \TempInfPrivReach{\TA}{\expiringBound}$, let $\timePriv{\varrun}$ be the time at
which $\varrun$ enters for the last time the private location. We denote by
$\VPriv{\varrun}$ the singleton $\{\timePriv{\varrun}\}$ if $\timePriv{\varrun}\in \setN$ and the open interval
$(\integralp{\timePriv{\varrun}},\integralp{\timePriv{\varrun}} +1)$ otherwise.
By definition of the region automaton, one can build runs going through the same path as~$\varrun$ in the region automaton
of $\TA$ but reaching the private location at any point within $\VPriv{\varrun}$.
Similarly, given $\timeFinal{\varrun} = \runDuration{\varrun}$ the duration of $\varrun$ until the final location,
we denote $\VFinal{\varrun}$ the singleton $\{\timeFinal{\varrun}\}$ if $\timeFinal{\varrun}\in \setN$ and the open interval $(\integralp{\timeFinal{\varrun}},\integralp{\timeFinal{\varrun}} +1)$ otherwise. Let $\RRun{\varrun}$ be the set of runs
that follow the same path as $\varrun$ in the region automaton.
The set of durations of runs of $\RRun{\varrun}$ which belong to
$\TempInfPrivReach{\TA}{\expiringBound}$ is
$\VFinal{\varrun}\cap [0, \max_{\varrun'\in \RRun{\varrun},\runDuration{\varrun'}=\runDuration{\varrun}}(\VPriv{\varrun'}) + \expiringBound]$,
 which is either $\VFinal{\varrun}$ or
the left-open interval $(\integralp{\timeFinal{\varrun}}, \integralp{\timeFinal{\varrun}}+ \fract{\expiringBound}]$.
We denote by $\DPriv{\expiringBound}{\varrun}$ this set of durations.

Similarly, given a run $\varrun \in \TempSupPrivReach{\TA}{\expiringBound}$ reaching the final location at time $\timeFinal{\varrun}$,
we can again rely on the region automaton to build a set of durations
$\DPub{\expiringBound}{\varrun}$ describing the durations of runs that follow the same path as~$\varrun$ in the region automaton and that reach the final location more than $\expiringBound$ after entering the private location.
This set is of the form $\big\{\timeFinal{\varrun}\big\}$ if $\timeFinal{\varrun}\in\setN$,
or
$\big(\integralp{\timeFinal{\varrun}} + \fract{\expiringBound}, \integralp{\timeFinal{\varrun}} + 1\big)$
or
$\big(\integralp{\timeFinal{\varrun}}, \integralp{\timeFinal{\varrun}} + 1\big)$.

Assume first that $\TA$ is \fullTempOpaqueText{\expiringBound}.
As the set of durations reaching the final location is a union of intervals with integer bounds~\cite[Proposition 5.3]{BDR08} and as $\TA$ is \fullTempOpaqueText{\expiringBound}, the set
\TempInfPrivDurReach{\TA}{\expiringBound} and the set $\TempSupPrivDurReach{\TA}{\expiringBound} \cup \PubDurReach{\TA}$ describe the same union of intervals with integer bounds.
Let $t$ be a duration within those sets.
Then we will show that $t\in \TempInfPrivDurReach{\TA}{\expiringBound'} $ and
$t\in \TempSupPrivDurReach{\TA}{\expiringBound'} \cup \PubDurReach{\TA}$. Note that if
$t\in \PubDurReach{\TA}$ the latter statement is directly obtained, we will thus ignore this case in the following.
By definition of $\TempSupPrivDurReach{\TA}{\expiringBound}$ and $\TempInfPrivDurReach{\TA}{\expiringBound}$, there thus
exists a run $\runprive$ and a run $\runpublic$ such that
$t\in \DPriv{\expiringBound}{\runprive}$ and
$t\in \DPub{\expiringBound}{\runpublic}$.
Moreover, we can assume that those runs satisfy that $\DPriv{\expiringBound}{\runprive}$ and
$\DPub{\expiringBound}{\runpublic}$ do not depend on the bound $\expiringBound$ (\ie{} they are equal to $\VFinal{\runprive}$ and $\VFinal{\runpublic}$ respectively). Indeed, if such runs
did not exist, the set \TempInfPrivDurReach{\TA}{\expiringBound} or the set
$\TempSupPrivDurReach{\TA}{\expiringBound} \cup \PubDurReach{\TA}$ would have
$\integralp{t} + \fract{\expiringBound}$ as one of its bounds.
As a consequence, $\TempSupPrivDurReach{\TA}{\expiringBound}\cup \PubDurReach{\TA}=\TempSupPrivDurReach{\TA}{\expiringBound'}\cup \PubDurReach{\TA}$ and
$\TempInfPrivDurReach{\TA}{\expiringBound} = \TempInfPrivDurReach{\TA}{\expiringBound'}$.
Thus $\TA$ is \fullTempOpaqueText{\expiringBound'}.

Assume now that $\TA$ is \weakTempOpaqueText{\expiringBound}.
We consider first the case where $\expiringBound\geq \expiringBound'$.
There we have by definition $\TempInfPrivReach{\TA}{\expiringBound'}\subseteq \TempInfPrivReach{\TA}{\expiringBound}$ and $\TempSupPrivReach{\TA}{\expiringBound}\subseteq \TempSupPrivReach{\TA}{\expiringBound'}$, thus $\TA$ is \weakTempOpaqueText{\expiringBound'}.

Now assume that $\expiringBound< \expiringBound'$. The same reasoning as for the full version mostly applies.
As the set of durations reaching the final location is a union of intervals with integer bounds \cite[Proposition 5.3]{BDR08} and as $\TA$ is \weakTempOpaqueText{\expiringBound},
the set $\TempSupPrivDurReach{\TA}{\expiringBound} \cup \PubDurReach{\TA}$ describes the same union of intervals with integer bounds.
By the same reasoning as before, $\TempSupPrivDurReach{\TA}{\expiringBound}\cup \PubDurReach{\TA}=\TempSupPrivDurReach{\TA}{\expiringBound'}\cup \PubDurReach{\TA}$.
Moreover, given $t\in \TempInfPrivDurReach{\TA}{\expiringBound'}$, there exists $\runprive$
such that $t\in \DPriv{\expiringBound'}{\runprive}$.
Note that either $\DPriv{\expiringBound'}{\runprive} = \DPriv{\expiringBound}{\runprive}$ and is thus included in $\TempSupPrivDurReach{\TA}{\expiringBound} \cup \PubDurReach{\TA}$
or $\DPriv{\expiringBound'}{\runprive} = (\integralp{\timeFinal{\runprive}}, \integralp{\timeFinal{\runprive}} + \fract{\expiringBound'}]$ and $\DPriv{\expiringBound}{\runprive} = (\integralp{\timeFinal{\runprive}}, \integralp{\timeFinal{\runprive}} + \fract{\expiringBound}]$.
As the latter is included in $\TempSupPrivDurReach{\TA}{\expiringBound} \cup \PubDurReach{\TA}$ which only has integer bounds, then the former is included in it as well.

Remark that in the above~$\expiringBound$ and~$\expiringBound'$ can be freely swapped and thus $\TA$ is \weakfullTempOpaqueText{\expiringBound'} iff
$\TA$ is \weakfullTempOpaqueText{\expiringBound}.
\end{proof}

\section{Expiring execution-time opacity in PTAs}\label{section:PTAs}

We are now interested in the computation (and the emptiness) of the valuations set ensuring that a system is \weakfullTempOpaqueText{\expiringBound}.
We define the following problems, where we ask for parameter valuations $\pval$ and for valuations of~$\expiringBound$ \st{} $\valuate{\PTA}{\pval}$ is \weakfullTempOpaqueText{\expiringBound}.

\smallskip

\defProblem
{\weakfullTempOpacityText{\expiringBound} emptiness}
{A PTA~$\PTA$
}
{Decide whether the set of parameter valuations~$\pval$ and valuations of~$\expiringBound$ such that $\valuate{\PTA}{\pval}$ is \weakfullTempOpaqueText{\expiringBound} is empty}

\smallskip

\defProblem
{\weakfullTempOpacityText{\expiringBound} computation}
{A PTA~$\PTA$
}
{Synthesize the set of parameter valuations~$\pval$ and valuations of~$\expiringBound$ such that $\valuate{\PTA}{\pval}$ is \weakfullTempOpaqueText{\expiringBound}}

\begin{remark}\label{rmq:PTA-decision-problem}
A ``\fullTempOpacityText{\expiringBound} decision problem'' over PTAs is not defined; it aims to decide whether, given a parameter valuation $\pval$ and a bound~$\expiringBound$, a PTA is \fullTempOpaqueText{\expiringBound}:
it can directly reduce to the problem over a TA (which is decidable, \cref{th:TAtempop}). %
\end{remark}

\begin{example}\label{ex:opacityProblemsPTA}
	Consider again the PTA~$\PTA$ in~\cref{figure:example-PTA}.

	For this PTA, the answer to the \weakTempOpacityText{\expiringBound} emptiness problem is false, as there exists such a valuation (\eg{} the valuation given for~\cref{ex:opacityProblemsTA}).

	Moreover, we can show that, for all $\expiringBound$ and $\pval$:
	\begin{itemize}
		\item $\PubDurReach{\valuate{\PTA}{\pval}} = [0,3] $
		\item if $\valuate{\parami{1}}{\pval}>3$ or $\valuate{\parami{1}}{\pval}>\valuate{\parami{2}}{\pval}$, it is not possible to reach $\locfinal$ with a run passing through $\locpriv$ and therefore $\TempSupPrivDurReach{\valuate{\PTA}{\pval}}{\expiringBound} = \TempInfPrivDurReach{\valuate{\PTA}{\pval}}{\expiringBound} = \emptyset$
		\item if $\valuate{\parami{1}}{\pval}\leq3$ and $\valuate{\parami{1}}{\pval}\leq\valuate{\parami{2}}{\pval}$
		\begin{itemize}
			\item $\TempSupPrivDurReach{\valuate{\PTA}{\pval}}{\expiringBound} = (\valuate{\parami{1}}{\pval} + \expiringBound , \valuate{\parami{2}}{\pval}] $
			\item $\TempInfPrivDurReach{\valuate{\PTA}{\pval}}{\expiringBound} = [\valuate{\parami{1}}{\pval}, \min(\expiringBound+3, \valuate{\parami{2}}{\pval})] $
		\end{itemize}
	\end{itemize}

	Recall that the \fullTempOpacityText{\expiringBound} computation problem aims at synthesizing the valuations such that $\TempInfPrivDurReach{\valuate{\PTA}{\pval}}{\expiringBound} = \TempSupPrivDurReach{\valuate{\PTA}{\pval}}{\expiringBound} \cup \PubDurReach{\valuate{\PTA}{\pval}}$.
	The answer to this problem is therefore the set of valuations of timing parameters and of~$\expiringBound$ \st{} $\pval(\parami{1})=0 \wedge ((\expiringBound\leq 3 \wedge 3\leq \pval(\parami{2})\leq \expiringBound+3)\vee(\pval(\parami{2})<\expiringBound\wedge \pval(\parami{2})=3))$.
\end{example}
\subsection{The subclass of L/U-PTAs}
\begin{definition}[L/U-PTA~\cite{HRSV02}]\label{def:LUPTA}
	An \emph{L/U-PTA} is a PTA where the set of parameters is partitioned into lower-bound parameters and upper-bound parameters,
	where each upper-bound (resp.\ lower-bound) parameter~$\param_i$ must be such that,
	for every guard or invariant constraint $\clock \compOp \sum_{1 \leq i \leq \ParamCard} \alpha_i \param_i + d$, we have:
		${\compOp} \in \{ \leq, < \}$ implies $\alpha_i \geq 0$ (resp.\ $\alpha_i \leq 0$) and
		${\compOp} \in \{ \geq, > \}$ implies $\alpha_i \leq 0$ (resp.\ $\alpha_i \geq 0$).
\end{definition}
\begin{example}
	The PTA in \cref{figure:example-PTA} is an L/U-PTA with $\{ \param_1 \}$ as lower-bound parameter, and $\{ \param_2 \}$ as upper-bound parameter.
\end{example}

L/U-PTAs is the most well-known subclass of PTAs with some decidability results: for example, reachability-emptiness (``the emptiness of the valuations set for which a given location is reachable''), which is undecidable for PTAs, becomes decidable for L/U-PTAs~\cite{HRSV02}.
Various other results were studied (\eg{} \cite{BlT09,JLR15,ALR22}).
Concerning opacity, the existence of a parameter valuation \emph{and an execution time} such that the system is opaque is decidable for L/U-PTAs~\cite{ALMS22}, while the full ET-opacity emptiness becomes undecidable~\cite{ALMS22}.

Here, we show that both the \fullTempOpacityText{\expiringBound} emptiness and the \weakTempOpacityText{\expiringBound} emptiness problems are undecidable for L/U-PTAs.
This is both surprising (seeing from the existing decidability results for L/U-PTAs) and unsurprising, considering the undecidability of the full ET-opacity emptiness for this subclass~\cite{ALMS22}.

\begin{theorem}\label{thm:EmptinessLU}
	The \weakfullTempOpacityText{\expiringBound} emptiness problem is undecidable for L/U-PTAs with at least 4~clocks and 4~parameters.
\end{theorem}
\begin{figure}[tb]
	{\centering
		\resizebox{\columnwidth}{!}{
		\scriptsize
		\begin{tikzpicture}[->, >=stealth', auto, node distance=2cm, thin]

			\node[location] (l0) at (-1.25, 0) {$\locinit$};
			\node[cloud, cloud puffs=15.7, cloud ignores aspect, minimum width=4cm, minimum height=2cm, align=center, draw] (cloud) at (0, 0) {$\A$};
			\node[location] (lf) at (+1.25, 0) {$\locfinal$};

			\node[location, initial] (l0') at (-7, 0) {$\locinit'$};
			\node[location] (l1) at (-5, .7) {$\loc_1$};
			\node[location] (l2) at (-5, -.7) {$\loc_2$};
			\node[location] (l3) at (-3, 0) {$\loc_3$};
			\node[location, private] (lpriv) at (+3.0, 1.7) {$\locpriv$};

			\node[location] (l4) at (-1.25, 1.7) {$\loc_4$};

			\node[location, final] (lf') at (3, 0) {$\locfinal'$};

			\path
			(lpriv) edge node[left,align=center] {$\sclockx = 0$} (lf')
			(lf) edge node[xshift=1em] {$\sclockx = 2$} (lf')

			(l0') edge node[sloped,align=center] {$\styleparam{\param_1^l} \leq \sclockx \leq \styleparam{\param_1^u}$} (l1)
			(l1) edge node[sloped,align=center] {$\styleparam{\param_2^l} \leq \sclockx \leq \styleparam{\param_2^u}$} (l3)
			(l0') edge node[sloped,below] {$\styleparam{\param_2^l} \leq \sclockx \leq \styleparam{\param_2^u}$} (l2)
			(l2) edge node[sloped,align=center,below] {$\styleparam{\param_1^l} \leq \sclockx \leq \styleparam{\param_1^u}$} (l3)
			(l3) edge node[above,align=center,xshift=-1em] {$\sclockx = 1$} node[below,align=center, xshift=-1em] {$\styleclock{\Clock} \setminus \{ \sclockx \}$\\$ \leftarrow 0$} (l0)

			(l0') edge[out=90,in=165] node[align=center,below] {$\sclockx = 2$\\$\sclockx \leftarrow 0$} (lpriv)

			(l0') edge[out=90,in=180] node[below, xshift=3em] {$\bigvee_{i}(\styleparam{\param_i^l} < \sclockx \leq \styleparam{\param_i^u})$} (l4)
			(l4)  edge node[above] {$ \sclockx=1$} node[below] {$ \sclockx \leftarrow 0$} (lpriv)
			;

		\end{tikzpicture}
	}
	}
	\caption{\AuthorVersion{Construction for the undecidability}\FinalVersion{Undecidability} of \weakfullTempOpacityText{\expiringBound} emptiness for L/U-PTAs (used in~\cref{thm:EmptinessLU})}
	\label{figure:undecidabilityLU}
\end{figure}
\begin{proof}
	We reduce from the problem of reachability-emptiness in constant time, which is undecidable for general PTAs with at least 4~clocks and 2~parameters \cite[Lemma~7.1]{ALMS22}.
	That is, we showed that, given a constant time bound~$\boundReach$, the emptiness over the parameter valuations set for which a location is reachable in exactly~$\boundReach$ time units, is undecidable.

	Assume a PTA~$\PTA$ with 2~parameters, say $\param_1$ and $\param_2$, and a target location~$\locfinal$.
	Fix $\boundReach = 1$.
	From~\cite[Lemma~7.1]{ALMS22}, it is undecidable whether there exists a parameter valuation for which~$\locfinal$ is reachable in time~$1$.

	The idea of our proof is that, as in~\cite{JLR15,ALMS22}, we ``split'' each of the two parameters used in~$\PTA$ into a lower-bound parameter ($\param_1^l$ and~$\param_2^l$) and an upper-bound parameter ($\param_1^u$ and~$\param_2^u$).
	Each construction of the form $\clock < \param_i$ (resp.\ $\clock \leq \param_i$) is replaced with $\clock < \param_i^u$ (resp.\ $\clock \leq \param_i^u$)
	while
	each construction of the form $\clock > \param_i$ (resp.\ $\clock \geq \param_i$) is replaced with $\clock > \param_i^l$ (resp.\ $\clock \geq \param_i^l$);
	$\clock = \param_i$ is replaced with $\param_i^l \leq \clock \leq \param_i^u$.
	Therefore, the PTA~$\PTA$ is exactly equivalent to our construction with duplicated parameters, provided $\param_1^l = \param_1^u$ and $\param_2^l = \param_2^u$.
	The crux of the rest of this proof is that we will ``rule out'' any parameter valuation not satisfying these equalities, so as to use directly the undecidability result of \cite[Lemma~7.1]{ALMS22}.

	Consider the extension~$\PTA'$ of~$\PTA$ given in \cref{figure:undecidabilityLU}, containing notably new locations $\locinit'$, $\locpriv$, $\locfinal'$, $\loc_i$ for $i=1,\cdots,4$, and a number of guards as seen on the figure; we assume that $\clock$ is an extra clock not used in~$\PTA$.
	The guard on the transition %
		from $\locinit'$ to~$\loc_4$ %
		stands for 2 different transitions guarded with
	$\param_1^l < \clock \leq \param_1^u$,
	and
	$\param_2^l < \clock \leq \param_2^u$,
	respectively.

	Due to the fact that $\locpriv$ must be exited in 0-time to reach~$\locfinal'$, note that, for any $\expiringBound$, the system is \weakfullTempOpaqueText{\expiringBound} iff it is \weakfullTempOpaqueText{0}.

	Let us first make the following observations, for any parameter valuation $\pval'$:
	\begin{enumerate}
		\item one can only take the upper most transition directly from~$\locinit'$ to~$\locpriv$ at time~2, \ie{} $\locfinal'$ is always reachable in time~2 via a run visiting location~$\locpriv$: $2\in\PrivDurReach{\valuate{\PTA'}{\pval'}}$;
		\item the original PTA~$\PTA$ can only be entered whenever $\param_1^l \leq \param_1^u$ and $\param_2^l \leq \param_2^u$; going from~$\locinit'$ to~$\locinit$ takes exactly 1 time unit (due to the $\clock = 1$ guard);

		\item if $\locfinal'$ is reachable by a public run (not passing through $\locpriv$), then its duration is
			necessarily exactly~2 (going through~$\PTA$);

		\item we have $ \TempSupPrivDurReach{\valuate{\PTA'}{\pval'}}{0} = \emptyset$ as any run reaching~$\locfinal'$ and visiting $\locpriv$ can only do it immediately;

		\item from \cite[Lemma 7.1]{ALMS22}, it is undecidable whether there exists a parameter valuation for which there exists a run reaching~$\locfinal$ from~$\locinit$ in time $1$, \ie{} reaching~$\locfinal'$ from~$\locinit'$ in time~$2$.
	\end{enumerate}

	Let us consider the following cases.
	\begin{enumerate}
		\item If $\param_1^l > \param_1^u$ or $\param_2^l > \param_2^u$, then due to the guards from~$\locinit'$ to~$\locinit$, %
			there is no way to reach~$\locfinal'$ with a public run; since $\locfinal'$ can still be reached for some execution times (notably ${\clock = 2}$ through the upper transition from~$\locinit'$ to~$\locpriv$), then $\PTA'$ cannot be \weakfullTempOpaqueText{0}.

		\item If $\param_1^l < \param_1^u$ or $\param_2^l < \param_2^u$, then
			one of the transitions from~$\locinit'$ to~$\loc_4$ can be taken, and $\TempInfPrivDurReach{\valuate{\PTA'}{\pval'}}{0} = \set{1,2}$.
		Moreover, $\locfinal'$ might only be reached by a public run of duration~$2$ through~$\PTA$.
		Therefore, $\PubDurReach{\valuate{\PTA'}{\pval'}}\subseteq [2,2]$.
		Therefore $\PTA'$ cannot be \weakfullTempOpaqueText{0} for any of these valuations.

		\item If $\param_1^l = \param_1^u$ and $\param_2^l = \param_2^u$, then the behavior of the modified~$\PTA$ (with duplicate parameters) is exactly the one of the original~$\PTA$.
		Also, note that the transition from~$\locinit'$ to~$\locfinal'$ via~$\loc_4$ cannot be taken.
		In contrast, the upper transition from~$\locinit'$ to~$\locpriv$ can still be taken.

		Now, assume there exists a parameter valuation for which there exists a run of~$\PTA$ of duration $1$ reaching $\locfinal$.
		And, as a consequence, $\locfinal'$ is reachable, and therefore there exists some run of duration~2 (including the 1 time unit to go from~$\locinit$ to~$\locinit'$) reaching $\locfinal'$ after passing through~$\PTA$, which is public.
		From the above reasoning, all runs reaching~$\locfinal'$ have duration~2; in addition, we exhibited a public and a secret run; therefore the modified automaton $\PTA'$ is \weakfullTempOpaqueText{0} for such a parameter valuation.

		Conversely, assume there exists no parameter valuation for which there exists a run of~$\PTA$ of duration $1$ reaching~$\locfinal$.
		In that case, $\PTA'$ is not \weakfullTempOpaqueText{0} for any parameter valuation: $\TempInfPrivDurReach{\valuate{\PTA'}{\pval'}}{0} = [2,2]$ and $2\not\in \TempSupPrivDurReach{\valuate{\PTA'}{\pval'}}{0} \cup \PubDurReach{\valuate{\PTA'}{\pval'}} = \emptyset$).
	\end{enumerate}

	As a consequence, there exists a parameter valuation~$\pval'$ for which $\valuate{\PTA'}{\pval'}$ is \weakfullTempOpaqueText{\expiringBound} iff there exists a parameter valuation~$\pval$ for which there exists a run in~$\valuate{\PTA}{\pval}$ of duration $1$ reaching $\locfinal$---which is undecidable from \cite[Lemma 7.1]{ALMS22}.

	 The undecidability of the reachability-emptiness in constant time for PTAs holds from 4 clocks and 2 parameters~\cite[Lemma~7.1]{ALMS22}.
	 Here, we duplicate the parameters (which gives 4~parameters), and we add a fresh clock~$\clock$, never reset (except from~$\locpriv$ to~$\locfinal'$);
	 however, the construction of~\cite[Lemma~7.1]{ALMS22} also uses a special clock never reset.
	 Since ours is only reset ``after'' the original~$\PTA$, we can reuse the same clock.
	 Therefore, our result holds from 4~clocks and 4~parameters.
\end{proof}

As the emptiness problems are undecidable, the computation problems are immediately intractable as well.

\begin{corollary}\label{thm:computationLU}
	The \weakfullTempOpacityText{\expiringBound} computation problem is unsolvable for L/U-PTAs with at least 4~clocks and 4~parameters.
\end{corollary}
\subsection{The full class of PTAs}

The undecidability of the emptiness problems for L/U-PTAs proved above immediately implies undecidability for the larger class of PTAs.
However, we provide below an original proof, with a smaller number of parameters.

\begin{figure}[tb]
	{\centering
		\scriptsize
		\begin{tikzpicture}[pta, node distance=2cm, thin]

			\node[location] (l0) at (-2, 0) {$\locinit$};
			\node[location] (lf) at (+1.8, 0) {$\locfinal$};
			\node[cloud, cloud puffs=15.7, cloud ignores aspect, minimum width=5cm, minimum height=2cm, align=center, draw] (cloud) at (0cm, 0cm) {$\PTA$};

			\node[location, initial] (l0') at (-3.5, 0) {$\locinit'$};
			\node[location, private] (lpriv) at (+2, -1.2) {$\locpriv$};
			\node[location, final] (lf') at (+3.5, 0) {$\locfinal'$};

			\path
			(l0') edge node[above, xshift=-.5em]{$\sclockx = 0$} (l0) %
			(lf) edge node[above, xshift=+.9em]{$\sclockx = 1$} (lf')
			;
			\path (l0') edge[out=-45,in=180] node[above] {$\sclockx = 1$} node[below,align=center]{$\sclockx \assign 0$} (lpriv);
			\path (lpriv) edge[out=0,in=-90] node[below right,align=center]{$\sclockx = 0$} (lf');

		\end{tikzpicture}
	}
	\caption{\AuthorVersion{Construction for the undecidability}\FinalVersion{Undecidability} of \weakfullTempOpacityText{\expiringBound} emptiness for PTAs (used in~\cref{thm:EmptinessPTA})}
	\label{figure:undecidability}
\end{figure}
\begin{theorem}\label{thm:EmptinessPTA}
	The \weakfullTempOpacityText{\expiringBound} emptiness problem is undecidable for general PTAs for at least 4~clocks and 2~parameters.
\end{theorem}
\begin{proof}
	We reduce again from the problem of reachability-emptiness in constant time, which is undecidable for general PTAs with at least 4~clocks and 2~parameters \cite[Lemma~7.1]{ALMS22}.

	Fix $\boundReach = 1$.
	Consider an arbitrary PTA~$\PTA$, with initial location~$\locinit$ and a given location~$\locfinal$.
	We add to~$\PTA$ a new clock~$\clock$ (unused and therefore never reset in~$\PTA$), and we add the following locations and transitions in order to obtain a PTA~$\PTA'$, as in \cref{figure:undecidability}:
		a new initial location $\locinit'$, with an urgent outgoing transition to~$\locinit$, and a transition to a new location~$\locpriv$ enabled after 1~time unit;
		a new final location~$\locfinal'$ with incoming transitions from $\locpriv$ (in 0-time) and from $\locfinal$ (after 1 time unit since the system start).
	First, due to the guard ``$\clock = 0$'' from~$\locpriv$ to~$\locfinal'$, note that, for any $\expiringBound$, the system is \weakfullTempOpaqueText{\expiringBound} iff it is \weakfullTempOpaqueText{0}.
	Also note that, for any valuation, $\TempInfPrivDurReach{\valuate{\PTA'}{\pval}}{0} = [1, 1]$.
	For the same reason, note that $\TempSupPrivDurReach{\valuate{\PTA'}{\pval}}{0} = \emptyset$.
	Second, note that, due to the guard ``$\clock = 1$'' on the edge from~$\locfinal$ and~$\locfinal'$ (with $\clock$ never reset along this path), $\PubDurReach{\valuate{\PTA'}{\pval}}$ can at most contain $[1, 1]$, \ie{} $\PubDurReach{\valuate{\PTA'}{\pval}} \subseteq [1,1]$.

	Now, let us show that there exists a valuation $\pval$ such that $\valuate{\PTA'}{\pval}$ is \weakfullTempOpaqueText{0} iff there exists $\pval$ such that $\locfinal$ is reachable in $\valuate{\PTA}{\pval}$ in 1~time unit.

	\begin{itemize}
		\item[$\Rightarrow$]
		Assume there exists a valuation~$\pval$ such that $\valuate{\PTA'}{\pval}$ is \weakfullTempOpaqueText{0}.

		Recall that, from the construction of~$\PTA'$, $\TempInfPrivDurReach{\valuate{\PTA'}{\pval}}{0} = [1, 1]$.
		Therefore, from the definition of \weakfullTempOpacityText{0}, there exist runs only of duration~1 (resp.\ there exists at least a run of duration~1) reaching $\locfinal'$ without visiting~$\locpriv$.
		Since $\PubDurReach{\valuate{\PTA'}{\pval}} \subseteq [1,1]$, then $\locfinal$ is reachable in exactly 1 time unit in $\valuate{\PTA}{\pval}$.

		\item[$\Leftarrow$]
		Assume there exists~$\pval$ such that~$\locfinal$ is reachable in $\valuate{\PTA}{\pval}$ in exactly 1 time unit.
		Therefore, $\locfinal'$ can also be reached in exactly 1 time unit: hence, $\PubDurReach{\valuate{\PTA'}{\pval}} = [1,1]$.

		Now, recall that $\TempSupPrivDurReach{\valuate{\PTA'}{\pval}}{0} = \emptyset$ and $\TempInfPrivDurReach{\valuate{\PTA'}{\pval}}{0} = [1, 1]$.
		Therefore, $\TempInfPrivDurReach{\valuate{\PTA'}{\pval}}{0} = \TempSupPrivDurReach{\valuate{\PTA'}{\pval}}{0} \cup \PubDurReach{\valuate{\PTA'}{\pval}}$, which from \cref{definition:temporary-timed-opacity} means that $\valuate{\PTA'}{\pval}$ is \fullTempOpaqueText{0}.
		Trivially, we also have that $\TempInfPrivDurReach{\valuate{\PTA'}{\pval}}{0} \subseteq \TempSupPrivDurReach{\valuate{\PTA'}{\pval}}{0} \cup \PubDurReach{\valuate{\PTA'}{\pval}}$ and therefore $\valuate{\PTA'}{\pval}$ is also \weakTempOpaqueText{0}.
	\end{itemize}
	Therefore, there exists~$\pval$ such that $\valuate{\PTA'}{\pval}$ is \weakfullTempOpaqueText{0} iff $\locfinal$ is reachable in $\valuate{\PTA}{\pval}$ in 1~time unit---which is undecidable \cite[Lemma~7.1]{ALMS22}.
	As a conclusion, \weakfullTempOpacityText{\expiringBound} emptiness is undecidable.

	Concerning the number of clocks and parameters, we use the same argument as in the proof of \cref{thm:EmptinessLU}\AuthorVersion{: the undecidability of the reachability-emptiness in constant time holds from 4 clocks and 2 parameters, and we add a fresh clock~$\clock$, but which can be shared with the global clock of \cite[Lemma~7.1]{ALMS22}.
	Therefore, our construction requires 4~clocks and 2~parameters}.
\end{proof}
\begin{corollary}\label{thm:computationPTA}
	The \weakfullTempOpacityText{\expiringBound} computation problem is unsolvable for PTAs\AuthorVersion{ for at least 4~clocks and 2~parameters}.
\end{corollary}
\section{Conclusion and perspectives}\label{section:conclusion}

\AuthorVersion{\paragraph{Conclusion}}
We studied here a version of execution-time opacity where the secret has an expiration date: that is, we are interested in computing the set of expiration dates of the secret for which the attacker is unable to deduce whether the secret was visited \emph{recently} (\ie{} before its expiration date) prior to the system completion; the attacker has access only to the model and to the execution time of the system.
We considered both the full opacity (the system must be opaque for all execution times) and the weak opacity (the set of execution times visiting the secret before its expiration date is included into the set of execution times reaching the final location).
Given a known constant expiration date, the decision problems are all decidable for \AuthorVersion{timed automata}\FinalVersion{TAs}; in addition, we can effectively \emph{compute} the set of expiration dates for which the system is \emph{weakly} opaque (full opacity remains open).
However, parametric versions of these problems, with unknown timing parameters, turned to be all undecidable, including for the L/U-PTA subclass of PTAs\AuthorVersion{, previously known for some decidability results}.
\AuthorVersion{This shows the hardness of the considered problem.}

\paragraph{Summary}
We summarize our results in \cref{tab:resultsRecall}.
``\cellYes{}'' denotes decidability, while ``\cellNo{}'' denotes undecidability; ``\cellOpen{}'' denotes an open problem.

\begin{table}
	\centering
	\caption{Summary of the results}
	\begin{tabular}{|l|l|l|l|l|}
		\hline
		\multicolumn{2}{|l|}{}          & \cellHeader{Decision}                        & \cellHeader{Emptiness}                   & \cellHeader{Computation}              \\ \hline
		\multirow{2}{*}{TA}      & Weak & \cellYes \cellCref{th:TAtempop}              & \cellYes  \cellCref{th:weakEmptiness}    & \cellYes \cellCref{th:weakcomp}       \\
		                         & Full & \cellYes \cellCref{th:TAtempop}            & \cellYes \cellCref{thm:opacityEmptiness} & \cellOpen{}                           \\ \hline
		\multirow{2}{*}{L/U-PTA} & Weak & \cellYes \cellCref{rmq:PTA-decision-problem} & \cellNo \cellCref{thm:EmptinessLU}       & \cellNo \cellCref{thm:computationLU}  \\
		                         & Full & \cellYes \cellCref{rmq:PTA-decision-problem} & \cellNo \cellCref{thm:EmptinessLU}       & \cellNo \cellCref{thm:computationLU}  \\ \hline
		\multirow{2}{*}{PTA}     & Weak & \cellYes \cellCref{rmq:PTA-decision-problem} & \cellNo \cellCref{thm:EmptinessPTA}      & \cellNo \cellCref{thm:computationPTA} \\
		                         & Full & \cellYes \cellCref{rmq:PTA-decision-problem} & \cellNo \cellCref{thm:EmptinessPTA}      & \cellNo \cellCref{thm:computationPTA} \\ \hline
	\end{tabular}
	\label{tab:resultsRecall}
\end{table}

\paragraph{Perspectives}
The main theoretical future work is the open problem in \cref{tab:resultsRecall} (\fullTempOpacityText{\expiringBound} computation)\AuthorVersion{: it is unclear whether we can \emph{compute} the exact set of expiration dates~$\expiringBound$ for which a system is \fullTempOpaqueText{\expiringBound}}.

The proofs of undecidability in \cref{section:PTAs} require a minimal number of clocks and parameters.
Smaller numbers might lead to decidability.
In addition, the same proofs are based on an undecidability result (reachability emptiness in constant time~\cite[Lemma~7.1]{ALMS22})\ which uses \emph{rational}-valued parameters.
The undecidability of the emptiness problems of \cref{section:PTAs} over \emph{integer}-valued parameters\AuthorVersion{ does not follow immediately, and} remains to be shown.

While the non-parametric part can be (manually) encoded into existing problems~\cite{ALMS22} using a TA transformation in order to reuse our implementation in \imitator{}~\cite{Andre21}, the implementation of the parametric problems remains to be done.
Since the emptiness problem is undecidable, this implementation can only come in the form of a procedure without a guarantee of termination, or with an approximate result.
In addition to weak and full ET-opacity,
problems focusing on the opacity for \emph{at least} one execution time might give a different decidability or complexity; for example, we highly suspect that the complexity of \cref{th:TAtempop} would decrease in this latter situation.

\newcommand{\CCIS}{Communications in Computer and Information Science}
	\newcommand{\ENTCS}{Electronic Notes in Theoretical Computer Science}
	\newcommand{\FI}{Fundamenta Informormaticae}
	\newcommand{\FundInf}{Fundamenta Informormaticae}
	\newcommand{\FMSD}{Formal Methods in System Design}
	\newcommand{\IJFCS}{International Journal of Foundations of Computer Science}
	\newcommand{\IJSSE}{International Journal of Secure Software Engineering}
	\newcommand{\IPL}{Information Processing Letters}
	\newcommand{\JLAP}{Journal of Logic and Algebraic Programming}
	\newcommand{\JLC}{Journal of Logic and Computation}
	\newcommand{\LMCS}{Logical Methods in Computer Science}
	\newcommand{\LNCS}{Lecture Notes in Computer Science}
	\newcommand{\RESS}{Reliability Engineering \& System Safety}
	\newcommand{\STTT}{International Journal on Software Tools for Technology Transfer}
	\newcommand{\TOSEM}{ACM Transactions on Software Engineering and Methodology}
	\newcommand{\TCS}{Theoretical Computer Science}
	\newcommand{\ToPNoC}{Transactions on Petri Nets and Other Models of Concurrency}
	\newcommand{\TSE}{IEEE Transactions on Software Engineering}

\AuthorVersion{%
	\renewcommand*{\bibfont}{\small}
	\printbibliography[title={References}]
}

\end{document}